\documentclass[journal,10pt]{IEEEtran}

\pdfoutput=1

\makeatletter
\def\ps@headings{%
\def\@oddhead{\mbox{}\scriptsize\rightmark \hfil \thepage}%
\def\@evenhead{\scriptsize\thepage \hfil \leftmark\mbox{}}%
\def\@oddfoot{}%
\def\@evenfoot{}}
\makeatother
\usepackage{color}
\usepackage[cmex10]{amsmath}
\usepackage{comment}
\usepackage[pdftex]{graphicx}
\usepackage{bm}
\usepackage{amsthm}
\usepackage{amssymb}
\usepackage{enumerate}
\usepackage{stfloats}
\usepackage{cases}
\usepackage{epstopdf}
\usepackage{thmtools}
\usepackage{graphicx}
\usepackage{mathtools}
\usepackage[utf8]{inputenc}
\usepackage{soul} 
\usepackage{hyperref}
\usepackage{pgfplots}
\usepackage{etoolbox}
\usepackage{makecell}
\usepackage{algorithm}
\usepackage{algorithmic}
\usepackage{tikz}
\usetikzlibrary{shapes.geometric, arrows, positioning, matrix, patterns}

\theoremstyle{definition}

\AfterEndEnvironment{problem}{\noindent\ignorespaces}

\AfterEndEnvironment{lemma}{\noindent\ignorespaces}

\AfterEndEnvironment{definition}{\noindent\ignorespaces}
\newtheorem {theorem}{Theorem}
\AfterEndEnvironment{theorem}{\noindent\ignorespaces}
\newtheorem {proposition}{Proposition}
\AfterEndEnvironment{proposition}{\noindent\ignorespaces}

\AfterEndEnvironment{corollary}{\noindent\ignorespaces}

\begin{document}


\title{\huge{Energy Efficiency Maximization Via Joint Sub-Carrier Assignment and Power Control for OFDMA Full Duplex Networks}}


\author{Rojin~Aslani, \emph{Student Member, IEEE}, Mehdi~Rasti, \emph{Member, IEEE,} and Ata~Khalili, \emph{Student Member, IEEE}
\thanks{Copyright (c) 2015 IEEE. Personal use of this material is permitted. However, permission to use this material for any other purposes must be obtained from the IEEE by sending a request to pubs-permissions@ieee.org.}
\thanks{The authors are with the Department of Computer Engineering and Information Technology, Amirkabir University of Technology, Tehran, Iran (email: \{rojinaslani, rasti\}@aut.ac.ir and ata.khalili@ieee.org.)}
}

\maketitle

\begin{abstract}
In this paper, we develop an energy efficient resource allocation scheme for orthogonal frequency division {multiple access} (OFDMA) networks with in-band full-duplex (IBFD) {communication between} the base station and user equipments (UEs) considering a realistic self-interference (SI) model. 
Our primary aim is {to maximize} the system energy efficiency (EE) through a joint power control and sub-carrier assignment in both {the} downlink (DL) and uplink (UL), {where} the quality of service requirements of {the} UEs in DL and UL {are guaranteed}. The formulated problem is non-convex due to the non-linear fractional objective function and {the} non-convex feasible set which is generally intractable. In order to handle this difficulty, {we first use fractional programming to transform the fractional objective function to the subtractive form. Then, by employing Dinkelbach method, we propose an iterative algorithm in which an inner problem is solved in each iteration. Applying majorization-minimization approximation, we make the inner problem convex. Also, by introducing a penalty function to handle integer sub-carrier assignment variables, we propose an iterative algorithm for addressing the inner problem. We show that our proposed algorithm converges to the locally optimal solution which is also demonstrated by our simulation results.} In addition, simulation results show that by applying {the} IBFD capability in OFDMA networks with efficient SI cancellation techniques, our proposed resource allocation algorithm attains a 75\% {increase in the EE} {as} compared to {the} half-duplex system.
\end{abstract}


\IEEEpeerreviewmaketitle

\section{Introduction}
Nowadays, the fast development of wireless communication technologies increases energy consumption and carbon emission which raises concerns {across the globe}. According to \cite{carbonEmission}, it is estimated that the percentage of the global carbon emission due to the information and communication technologies is 5\% which increases significantly in coming years, and the situation will intensify with the arrival of 5G networks in the near future. Moreover, it is reported that network operators spend more than 10 billion dollars a year on electricity \cite{EnergyExpenditure}. Therefore, due to {the} importance of energy consumption in terms of environmental impact and cost, energy efficiency (EE) will be a significant feature in 5G wireless networks \cite{ekram}, {on which we focus} in this paper.

Recently, there have been some research efforts to design energy efficient resource allocation schemes for {o}rthogonal {f}requency {d}ivision {multiple} {a}ccess (OFDMA) cellular network{s} \cite{Venturino}-\cite{Miao}. The authors of \cite{Venturino} addressed the joint problem of user scheduling and power control to maximize the EE in the downlink (DL) of OFDMA cellular networks. In \cite{He}, an energy efficient resource allocation algorithm with proportional fairness is developed for DL multi-user OFDMA systems with distributed antennas. A low-complexity sub-optimal algorithm is developed in \cite{He}, which allocates sub-carrier{s} and power for the EE maximization considering constraints of bit error rates, proportional fair data rates and total transmit power of remote access units. Applying the Lagrangian dual decomposition technique, an energy efficient sub-carrier assignment and power control {scheme} is proposed in \cite{Xiao}, which aims at maximizing {the} EE. In \cite{loodaricheh}, the authors designed the joint relay selection, pairing, sub-carrier assignment and power control algorithms for maximizing the system EE {with} quality of service (QoS) {considerations}. Dinkelbach method is employed in \cite{loodaricheh} to tackle the non-linear fractional objective function and an optimal solution to relaxed problem is presented. In \cite{Ng}-\cite{Dong}, resource allocation schemes are proposed for EE maximization in the DL of OFDMA cellular networks with energy harvesting capability for {the} base station (BS) or user equipments (UEs). The resource allocation schemes to maximize the EE in the uplink (UL) of OFDMA wireless networks are developed in \cite{Zappone}-\cite{Miao}.

\begin{table*}[ht]
\caption{Summary of Related Works and Comparison With Our Proposed Approach}
\label{related work}
\centering
\begin{tabular}{|c|c|c|c|c|c|}\hline
{\bf Ref.} & \thead{\bf Type of Comm.} & \thead{\bf Mode of Operation} & \thead{\bf Objective Function} & {\bf Constraints} & \thead{\bf Solution Approach} \\ 
\hline \hline
\cite{Venturino} & {DL} & \thead{HD for BS} & \thead{Maximizing EE} & Feasible transmit power of BS & \thead{Centralized} \\
\hline
\cite{He} & {DL} & \thead{HD for BS} & \thead{Maximizing EE} &  \thead{Exclusive sub-carrier assignment, \\ Feasible transmit power of BS, \\ QoS constraint for UEs} & Centralized \\
\hline
\cite{Xiao} & {DL} & \thead{HD for BS} & \thead{Maximizing EE} & \thead{Exclusive sub-carrier assignment, \\ Feasible transmit power of BS} & \thead{Centralized, \\ Optimal \& sub-optimal} \\
\hline
\cite{loodaricheh} & \thead{UL, \\ DL} & \thead{OBFD for BS, \\ HD for UEs} & \thead{Maximizing EE} & \thead{Exclusive sub-carrier assignment, \\ Feasible transmit power of BS, \\ QoS constraint for UEs} & \thead{Centralized, \\ Dinkelbach method } \\
\hline
\cite{Dong} & DL & \thead{IBFD for BS, \\ HD for UEs} & \thead{Maximizing EE} & \thead{Exclusive sub-carrier assignment, \\ Feasible transmit power of BS, \\ QoS constraint for UEs} & \thead{Centralized sub-optimal, \\ Distributed} \\
\hline
\cite{Zappone} & UL & HD for UEs & \thead{Maximizing EE} & \thead{Feasible transmit power of UEs, \\ QoS constraint for UEs} & \thead{Centralized, \\ Distributed} \\
\hline
{\cite{Schober}} & {DL} & {IBFD for Relays} & \thead{{Maximizing} \\ {sum{-}rate}} & \thead{{Exclusive sub-carrier assignment,} \\ {Feasible transmit power of BS \& Relays,} \\ {QoS constraint for UEs}} & \thead{{Distributed,} \\ {Dual} {decomposition}} \\
\hline
{\cite{Gang}} & {DL} & {IBFD for Relays} & \thead{{Maximizing EE}} & \thead{{Exclusive sub-carrier assignment,} \\ {Feasible transmit power of BS \& Relays,} \\ {QoS constraint for UEs}} & \thead{{Centralized,} \\ {Dinkelbach} {method}} \\
\hline
\cite{crick} & \thead{UL, \\ DL} & \thead{IBFD for BS, \\ HD for UEs} & \thead{Maximizing \\ sum{-}rate} & \thead{Exclusive sub-carrier assignment, \\ Feasible transmit power of BS \& UEs} & \thead{Distributed, \\ sub-optimal } \\
\hline
\cite{wen} & \thead{UL, \\ DL} & \thead{IBFD for BS, \\ HD for UEs} & \thead{Maximizing EE} & \thead{Exclusive sub-carrier assignment, \\ QoS constraint for UEs} & Centralized \\
\hline
{\cite{Power}} & {\thead{UL, \\ DL}} & {\thead{IBFD for BS, \\ HD for UEs}} & {\thead{\thead{Minimizing \\ power consumption}}} & {\thead{Feasible transmit power of BS \& UEs, \\ QoS constraint for UEs}} & {\thead{Centralized, \\ optimal \& sub-optimal}} \\
\hline
{\cite{Reverse}} & {\thead{UL, \\ DL}} & {\thead{IBFD for BS, \\ HD for UEs}} & {\thead{Maximizing \\  sum{-}rate}} & {\thead{Exclusive sub-carrier assignment, \\ Feasible transmit power of BS \& UEs}} & {\thead{Centralized, \\ optimal \& sub-optimal}} \\ 
\hline
\cite{nam} & \thead{UL, \\ DL} & \thead{IBFD for BS \& UEs, \\ Complete SIC} & \thead{Maximizing \\  sum{-}rate} & \thead{Exclusive sub-carrier assignment, \\ Feasible transmit power of BS \& UEs} & \thead{Distributed, \\ Local Pareto optimality } \\ 
\hline
\thead{Our \\ work} & \thead{UL, \\ DL} & \thead{IBFD for \\ BS \& UEs} & \thead{Maximizing EE} & \thead{Exclusive sub-carrier assignment, \\ Feasible transmit power of BS \& UEs, \\ QoS constraint for UEs} & \thead{Centralized, \\ Dinkelbach method} \\ 
\hline 
\end{tabular}
\end{table*} 

Most of the proposed resource allocation schemes in {the} literature for OFDMA cellular networks {are} devoted to either the DL or UL transmission. In fact, in {most works}, it is implicitly considered that the DL and UL channels operate in half-duplex (HD) mode or out-band full-duplex (OBFD) mode, {where} a radio transceiver can either transmit or receive at different time{s} on the same frequency band or on different frequency band{s} at the same time, respectively. Recent advances in signal processing techniques have challenged this presumption and indicated the practicability of in-band full-duplex (IBFD) {communication}, where a radio transceiver simultaneously transmits and receives on the same frequency band. But, the notion of concurrent transmission and reception in a node makes self-interference (SI) in IBFD systems which is {a portion} of the transmitted signal of {an} IBFD node received by itself, so interfering with the desired signal received at the same time. Therefore, the key requirement to implement the IBFD communication is applying self-interference cancellation (SIC) methods. There are variant SIC methods presented in the literature (e.g., see \cite{SIC1} and \cite{SIC2}). {Equipped with} SIC methods, the IBFD {communication} has attracted a growing interest from both industrial and academic world, due to its potential of doubling the spectral efficiency. 
However, there are few efforts for redesigning the resource allocation algorithms in IBFD cellular networks. The authors of \cite{Psomas}-\cite{nam} apply the IBFD capability to OFDMA wireless networks employing different architectures. In particular, IBFD cellular networks can be categorized into two-node and three-node architectures \cite{Psomas}. In two-node architecture, referred also as bidirectional, both nodes, i.e., the BS and the UEs have IBFD capability. However, in three-node architecture, only the BS is IBFD-capable and the UEs work in the HD mode. In \cite{Psomas}, the outage probability of an IBFD cellular network for both cases of the two-node and three-node architectures is analytically characterized.
In \cite{Schober}-\cite{Reverse}, it is assumed that the three-node architecture is employed in IBFD cellular networks. In \cite{Schober} and \cite{Gang}, resource allocation schemes are proposed for a relay assisted OFDMA DL cellular network in which the relays are IBFD-capable.
In \cite{crick}-\cite{nam}, resource allocation schemes are proposed for both the UL and DL of OFDMA networks with IBFD capability. A joint power control and sub-carrier assignment algorithm is proposed in \cite{crick} for the system {sum-rate} maximization subject to the maximum transmit power and sub-carrier assignment constraints. A greedy sub-carrier assignment algorithm and an iterative water-filling power control algorithm were proposed in \cite{crick}. The problem of the EE maximiziation is addressed in \cite{wen} in which two different SI models are considered: constant and linear SI model. Then, an optimal algorithm to achieve the maximum EE via a Lagrangian joint optimization of power control and sub-carrier assignment is developed for constant SI model. Also, by decoupling the problem into two sub-problems of power control and sub-carrier assignment, a heuristic algorithm is provided for linear SI model. In \cite{Power}, the problem of DL beamforming and antenna selection, alongside with UL power control with the goal of minimizing power consumption of the network is studied. In \cite{Reverse}, an optimal and a suboptimal joint power and sub-carrier allocation policies are proposed to maximize the weighted system sum-rate in a multicarrier non-orthogonal multiple access network. The two-node architecture for IBFD networks is considered in \cite{nam}, where it is assumed that SIC methods are able to cancel the SI almost completely, then the authors propose an iterative algorithm to jointly optimize the power control and sub-carrier assignment with the aim of maximizing the system sum-rate.

In this work, we focus on designing the energy efficient resource allocation scheme for joint power control and sub-carrier assignment in both {the} UL and DL {of} OFDMA networks with IBFD capability. Similar to \cite{nam}, we assume that the two-node architecture is employed in IBFD cellular network (opposing with \cite{crick}-\cite{Reverse} where the three-node architecture is employed). We further consider a more realistic assumption on the SI model in {the} BS and UEs. In fact, there are two main assumptions on the SI model in the literature: 1) {constant} SI independent of the transmit power \cite{wen}, \cite{nam}, which {has less complexity} but is not the case in practice, 2) {varying} SI proportional to the transmit power ({\cite{wen}, \cite{Reverse},} \cite{SImodel1}-\cite{SImodel3}, \cite{SICforBS}-\cite{SICforUE2}), {leading to the complex but realistic problem formulations,} which is adopted in our paper {for modeling the SI in both the BS and UEs.} Although the assumption {of} SI proportional to the transmit power is more realistic, it {invokes non-convexity, making} the {designing of} energy efficient resource allocation schemes more challenging in comparison with \cite{nam}. 
To the best of our knowledge, there is no work in the literature that proposes a {joint} energy efficient resource allocation scheme for both {the} UL and DL {of} OFDMA cellular network{s} with IBFD{-capable} BS and UEs, {and with a practical SI model that is} proportional to the transmit power (see Table \ref{related work}). The contributions of this work are summarized as follows:
\begin{itemize}
\item We present a {system model} for {the} OFDMA network with IBFD capability in which all UEs and {the} BS operate in {the} IBFD mode and {the} UL transmission from a {given} UE to the BS and {the} DL transmission from the BS to {that} UE occur simultaneously in the same frequency band. We suppose that the SI in both {the} BS and UEs are proportional to their transmit power and they have different SIC {capabilities as is the case in reality}. We also formally state the EE maximization problem subject to the maximum transmit power of {the} UEs and BS and the QoS requirement of {the} UEs {in the UL and DL}. To the best of our knowledge, this problem has not been considered in the literature for OFDMA cellular network{s} with IBFD capability for both {the} BS and UEs. 
\item We propose a joint solution to the sub-carrier assignment and power control in {the} UL{s} and DL{s} considering QoS provisioning. {To do this, we first use the
fractional programming to deal with the fractional objective function of the EE. Then, we apply Dinkelbach algorithm to address the problem in which an inner optimization problem should be solved in each iteration. Applying the majorization-minimization (MM) algorithm, we make the inner problem convex. Also, by introducing a penalty function, we handle the integer sub-carrier assignment variables. Finally, we solve the obtained convex optimization problem using off-the-shelf software packages, e.g., CVX.}
\item {Extensive simulation results show the convergence of our proposed algorithm to the locally optimal solution and outperforming other resource allocation schemes proposed in \cite{He}, \cite{wen}, \cite{nam} and \cite{Leng}. Also, our simulation results reveal the effect of several factors such as the minimum data rate requirement of UEs, the maximum transmit power, the cell size, and the SI on the EE of {the} IBFD system. Our numerical results demonstrate that by applying IBFD capability in OFDMA networks with efficient SIC techniques, our proposed resource allocation scheme can operate 75\% more energy efficiently than that of {the} HD system proposed in \cite{He}. Also, our proposed algorithm performs 20\% and 35\% better than the algorithms proposed in \cite{wen} to maximize EE in IBFD system{s} with linear and constant SI model, respectively.}
\end{itemize}

The summary of {the} related works and comparison with our proposed approach {is presented} in Table \ref{related work}. The rest of this paper is organized as follows. In Section II, we introduce the system model and formulate the problem. Our proposed method for addressing the stated problem is presented in Section III. The numerical results and conclusion are presented in Sections IV and V, respectively.

\section{System Model and Problem Formulation}

\begin{figure}[H]
\centering
\includegraphics[width=3.2in]{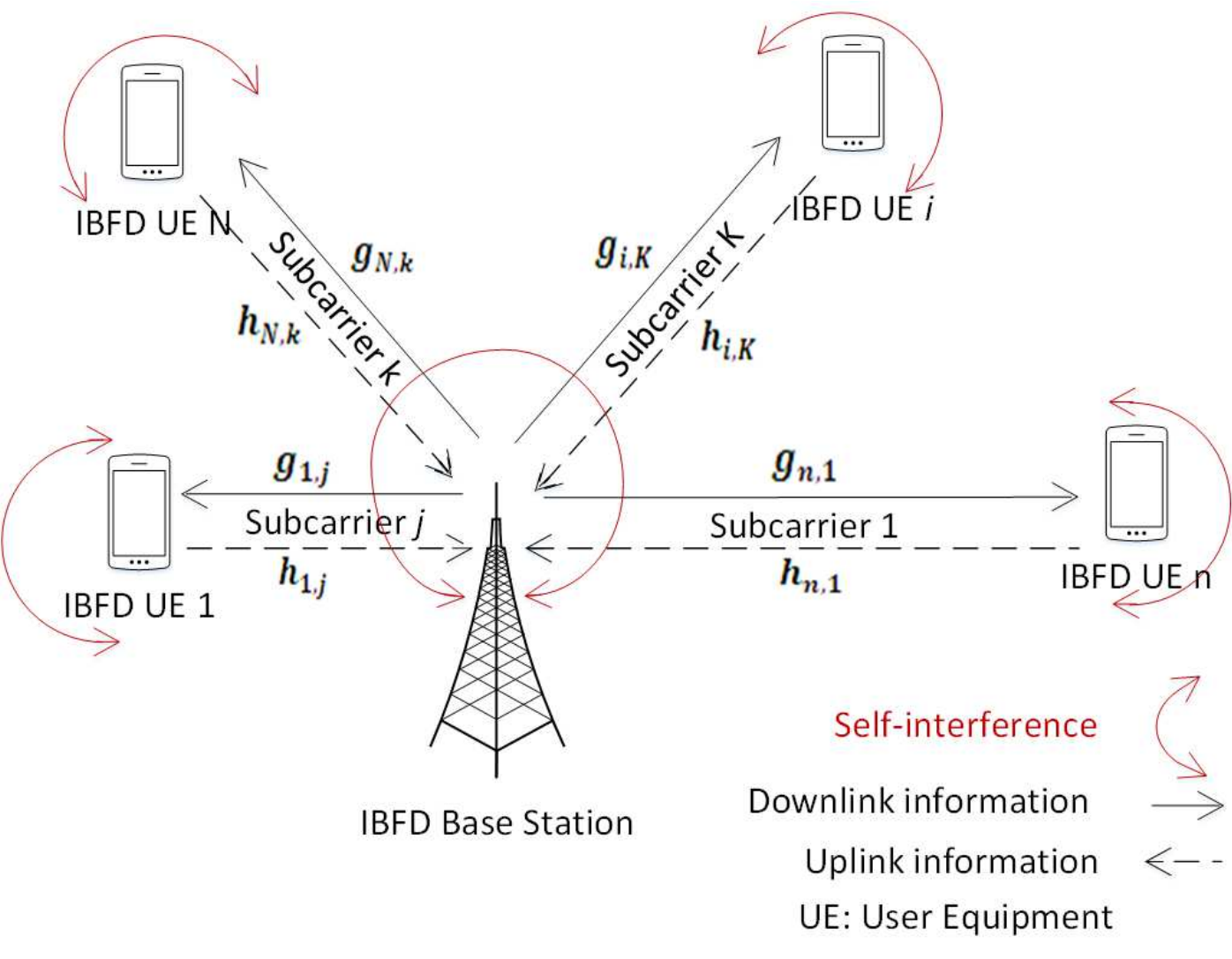}
\vspace{-10pt}
\caption{In-Band Full-Duplex OFDMA Wireless Network.}
\label{fig_systemModel}
\end{figure}

\subsection{System Model}

We consider a single-cell OFDMA network with one BS and $N$ UEs, the set of which is denoted by $\mathcal{N}\!\!=\!\!\{1,2,\cdots,N\}$. The system's total bandwidth is divided to $K$ sub-carriers perfectly orthogonal to each other, and let $\mathcal{K}\!=\!\!\{1,2,\cdots,K\}$ denote the set of sub-carriers. 

We assume that all the UEs and {the} BS operate in {the} IBFD mode and the UL and DL data transmission are simultaneously performed in the same frequency band. Assuming the separate antenna architecture for IBFD nodes presented in \cite{SIC1}, the BS and UEs are both equipped with two antennas in which the transmit and receive signals in the UL and DL use a dedicated antenna, both operating in the same frequency band.

Let $h_{n,k}$ denote the UL channel gain from the $n$th UE to the BS on {the} $k$th sub-carrier and $g_{n,k}$ denote the DL channel gain from the BS to {the} $n$th UE on {the} $k$th sub-carrier (see Fig. \ref{fig_systemModel}). Also, $p^\textnormal{u}_{n,k}$ and $p^{d}_{n,k}$ denote the transmit power of $n$th UE for transmitting {the} information signal to the BS on sub-carrier $k$ in {the} UL and the transmit power of the BS for transmitting {the} information to the $n$th UE on sub-carrier $k$ in {the} DL, respectively. The binary-valued sub-carrier assignment, $x_{n,k} \in \{0,1\}$ represents sub-carrier assignment for {the} $n$th UE on sub-carrier $k$ and is {defined as} 
\vspace{-3.5pt}
\begin{equation}
\label{subcarrierAssignment}
x_{n,k}=\begin{cases}
	1, & \textnormal{if {the} $k$th sub-carrier is assigned to {the} $n$th UE} \\
	0, & \textnormal{otherwise}. 
		\end{cases} 
\end{equation}

The vector $\bold{x}\in \mathbb{R}^{NK\times 1}$ represents the sub-carrier assignment of all the UEs and the vector $\bold{p}=[\bold{p}^\textnormal{u},\bold{p}^\textnormal{d}]^T$ represents the transmit power of UEs and BS, where $\bold{p}^\textnormal{u}=[p^\textnormal{u}_{n,k}]^T_{N\times K}$ is the UEs transmit power vector in UL and $\bold{p}^\textnormal{d}=[p^\textnormal{d}_{n,k}]^T_{N\times K}$ is the BS transmit power vector in DL. 
Considering a combination of the passive and active SIC methods \cite{Duarte} for all the UEs and BS, we assume that the residual SI is proportional to the transmit power (similar to \cite{Psomas}, \cite{Reverse}, \cite{SImodel1}-\cite{SImodel3} and \cite{SICforBS}-\cite{SICforUE2}). It should be noted that in practice, SI cannot be canceled completely even if the SI channel is perfectly known at the IBFD BS as well as IBFD UEs due to the limited dynamic range of the receiver \cite{FDMIMO}. Thus, the residual SI after cancellation at the receive antennas can be modeled as an independent zero-mean Gaussian distortion noise whose variance is proportional to the received power of the antenna \cite{Reverse}, \cite{FDMIMO}. Assuming different SIC capabilities in the BS and UEs, the SI power at the BS and the $n$th UE is represented by $\delta_\textnormal{BS} |l^\textnormal{SI}_\textnormal{BS}|^2 p^\textnormal{d}_{n,k}$ and $\delta_n |l^\textnormal{SI}_n|^2 p^\textnormal{u}_{n,k}$, respectively, where $0<\delta_\textnormal{BS} \ll 1$ and $0<\delta_n \ll 1$ are constants modeling the quality of the SIC at the BS and the $n$th UE, respectively, $l^\textnormal{SI}_\textnormal{BS} \in \bold{C}$ and $l^\textnormal{SI}_n \in \bold{C}$\footnote{{$\bold{C}$ denote the set of complex number.}} denote the SI channel gain at the BS and $n$th UE, respectively.
The noise is assumed to be additive white Gaussian noise (AWGN) whose power is $N_\textnormal{BS}$ at the BS and $N_n$ at the $n$th UE. In addition, we assume that global channel state information (CSI) of all channels is available at the BS\footnote{In fact, we assume that the BS attains the UL CSI by listening to the sounding reference signal transmitted by the UEs and the DL CSI through the channel quality indicator (CQI) feedback from the UEs \cite{nam}.} so as to unveil the performance upper bound of IBFD OFDMA wireless networks \cite{Reverse}. Given the sub-carrier assignment vector $\bold{x}$ and the transmit power vector $\bold{p}$, the UL SINR of {the} $n$th UE at the BS receiver on {the} $k$th sub-carrier denoted by $\gamma^\textnormal{u}_{n,k}$, is obtained {as}
\vspace{-5pt}
\begin{equation}
\label{gamma-UL}
\gamma^\textnormal{u}_{n,k}(\bold{x,p})=\frac{p^\textnormal{u}_{n,k} {x_{n,k}} h_{n,k}}{{\delta_\textnormal{BS} {|l^\textnormal{SI}_\textnormal{BS}|^2}} p^\textnormal{d}_{n,k} {x_{n,k}} + {N_\textnormal{BS}}}.
\end{equation}
Similarly, given the sub-carrier assignment vector $\bold{x}$ and the transmit power vector $\bold{p}$, the DL SINR of {the} $n$th UE on {the} $k$th sub-carrier denoted by $\gamma^\textnormal{d}_{n,k}$, is obtained {as}
\vspace{-2pt}
\begin{equation}
\label{gamma-DL}
\gamma^\textnormal{d}_{n,k}(\bold{x,p})=\frac{p^\textnormal{d}_{n,k} {x_{n,k}} g_{n,k}}{{\delta_n {|l^\textnormal{SI}_n|^2}} p^\textnormal{u}_{n,k} {x_{n,k}} + {N_n}}.
\end{equation}
According to the Shannon formula, the achievable instantaneous UL and DL transmission rate{s} (bps/Hz) for $n$th UE on $k$th sub-carrier are given by 
\begin{equation}
\label{UL-rate}
R^\textnormal{u}_{n,k}(\bold{x,p}) = \log \left(1+\gamma^\textnormal{u}_{n,k}(\bold{x,p})\right)\!,
\end{equation}
and
\vspace{-5pt}
\begin{equation}
\label{DL-rate}
R^\textnormal{d}_{n,k}(\bold{x,p}) = \log \left(1+\gamma^\textnormal{d}_{n,k}(\bold{x,p})\right)\!,
\end{equation}
respectively. 
Let $R^\textnormal{u}_n(\bold{x,p})$ denote the UL transmission rate of {the} $n$th UE, i.e., $R^\textnormal{u}_n(\bold{x,p})=\sum_{\forall k \in \mathcal{K}}{ R^\textnormal{u}_{n,k}(\bold{x,p}) } $ and let $R^\textnormal{d}_n(\bold{x,p})$ denote the DL transmission rate of {the} $n$th UE, i.e., $R^\textnormal{d}_n(\bold{x,p})=\sum_{\forall k \in \mathcal{K}}{ R^\textnormal{d}_{n,k}(\bold{x,p}) }$. The transmission rate of {the} $n$th UE is obtained by $R_n(\bold{x,p})= R^\textnormal{u}_n(\bold{x,p}) + R^\textnormal{d}_n(\bold{x,p})$. The total system sum-rate denoted by $R(\bold{x,p})$ is obtained by $R(\bold{x,p})$ $=\sum_{\forall n \in \mathcal{N}}{R_n(\bold{x,p})}$. 
The total system power consumption denoted by $P^\textnormal{T}(\bold{x},\bold{p})$ is {formed as}
\vspace{-2pt}
\begin{equation}
\label{totalPowerConsumption}
P^\textnormal{T}\!(\bold{x},\bold{p}) \!=\! P_\textnormal{BS}^\textnormal{C} +\!\! \sum_{\forall n \in \mathcal{N}}\!\!\!{P_n^\textnormal{C}} +\!\! \sum_{\forall n \in \mathcal{N}}\!{ \sum_{\forall k \in \mathcal{K}}\!\!{ x_{n,k} \!\left(\! {\frac{1}{\epsilon_n}} p^\textnormal{u}_{n,k} \!+\! {\frac{1}{\epsilon_\textnormal{BS}}} p^\textnormal{d}_{n,k}\!\right) } },
\end{equation}
where {$\epsilon_n \in (0,1)$ and $\epsilon_\textnormal{BS} \in (0,1)$ are the power amplifier efficiency of {the} $n$th UE and the BS, respectively \cite{powerModel}.} $P_\textnormal{BS}^\textnormal{C}$ and $P_n^\textnormal{C}$ are the fixed circuit power consumed by {the} BS and {the} $n$th UE, respectively. 
Indeed, we consider the transmit power of {the} BS and UEs as well as their circuit power consumption. We define the EE criterion as the ratio of the total system sum-rate to the total system power consumption given by 
\begin{equation}
\label{energyEfficiency}
EE(\bold{x},\bold{p})=\frac{R(\bold{x},\bold{p})}{P^\textnormal{T}(\bold{x},\bold{p})}.
\end{equation}

\subsection{Problem Formulation}

We formally state the joint optimization problem of sub-carrier assignment and power control to maximize the EE subject to constraints of the transmit power of both the UEs and BS and the QoS requirement{s} for each UEs at {the} DL and UL, that is:
\label{main problem} 
\begin{align}
\underset{\bold{p},\bold{x}}{\mathrm{maximize}}& \quad EE(\bold{x},\bold{p})  \label{problem}\\
\mathrm{subject\ to} & \quad \sum_{\forall n \in \mathcal{N}}{x_{n,k}} \leq 1, \qquad  \qquad \qquad \forall k \in \mathcal{K} \tag{8-1} \label{exclusive allocation of subcarriers} \\
& \quad \sum_{\forall k \in \mathcal{K}}{p^\textnormal{u}_{n,k} x_{n,k} }  \leq \overline{P}_n, \!  \qquad \qquad \forall n \in \mathcal{N} \tag{8-2} \label{transmit power constraint of users}\\
& \quad \sum_{\forall n \in \mathcal{N}}{ \sum_{\forall k \in \mathcal{K}}{p^\textnormal{d}_{n,k} x_{n,k} } } \leq \overline{P}_\textnormal{BS},  \tag{8-3} \label{transmit power constraint of base stations}\\
& \quad {R^\textnormal{u}_n(\bold{x,p}) \geq \overline{R}^\textnormal{u}_n, \qquad \qquad \qquad  \!\!\!\! \quad \!\!\!  \forall n \in \mathcal{N}} \tag{8-4} \label{target QoS constraint for each user at UL}\\
& \quad {R^\textnormal{d}_n(\bold{x,p}) \geq \overline{R}^\textnormal{d}_n, \qquad \qquad \qquad  \!\!\!\! \quad \!\!\!  \forall n \in \mathcal{N}} \tag{8-5} \label{target QoS constraint for each user at DL}\\
& \quad p^\textnormal{u}_{n,k} , p^\textnormal{d}_{n,k} \geq 0, \qquad \! \qquad  \forall n \in \mathcal{N},k \in \mathcal{K} \tag{8-6} \label{non-negative transmit power}\\
& \quad x_{n,k} \in \{0,1\}, \qquad \qquad   \forall n \in \mathcal{N},k \in \mathcal{K} \tag{8-7}, \label{binary subcarrier and base station assignment} \nonumber
\end{align}
where constraint (\ref{exclusive allocation of subcarriers}) ensures the exclusive sub-carrier assignment in OFDMA system. The feasibility of transmit power of {the} UEs and BS are indicated by constraints (\ref{transmit power constraint of users}) and (\ref{transmit power constraint of base stations}), respectively, in which $\overline{P}_n$ and $\overline{P}_\textnormal{BS}$ are the maximum transmit power of {the} $n$th UE and the BS, respectively. Constraints (\ref{target QoS constraint for each user at UL}) and (\ref{target QoS constraint for each user at DL}) guarantee the QoS requirement for each UEs at the UL and DL, respectively, where the specific required rates of the UEs in the UL and DL need to be satisfied.
Constraint (\ref{non-negative transmit power}) corresponds to the non-negative transmit power for {the} UEs and BS. Finally, constraint (\ref{binary subcarrier and base station assignment}) is an integer constraint for sub-carrier assignment to the UEs.

\section{Our Proposed Resource Allocation Scheme for EE Maximization}
Optimization problem (\ref{problem}) contains a non-convex objective function and non-linear constraints with a combination of binary and continuous variables, i.e., $\bold{x}$ and $\bold{p}$, respectively. Specifically, (\ref{problem}) is a non-convex mixed-integer nonlinear programming (MINLP) optimization problem and so it is hard to solve at {its} original form\footnote{To achieve the globally optimal solution of original problem (\ref{problem}), an exhaustive search is required{, entailing} a complexity of $O(N^K)$ which is computationally infeasible for ${N \gg 1}$ {and} $K \gg 1$.}. Thus we need to propose an efficient algorithm for addressing (\ref{problem}) with reasonable computational complexity. 

To evade the high complexity of the MINLP problem {in} (\ref{problem}), we reformulate it into a more mathematically tractable {one where} its time complexity is polynomial. The first concern to address the problem (\ref{problem}) is the coupled UL and DL power control and sub-carrier assignment variables in both the objective function and constraints. This makes obtaining a solution for problem (\ref{problem}) complicated. To tackle this issue, we define two new auxiliary transmit power variables and redefine the problem based on them. The next concern for addressing the problem is the fractional objective function of EE. 
We employ the fractional programming \cite{Dinkelbach} to transform the fractional objective function to the subtractive form. Then, we employ Dinkelbach algorithm to address the problem in subtractive form of the objective function. In each iteration of Dinkelbach algorithm, an inner optimization problem requires to be solved. To address the inner problem, we first make it convex by applying MM algorithm \cite{Sun}. Next, we treat the integer variables and constraints on the sub-carrier assignment variables which create a disjoint feasible solution set {that} is an obstacle to solve the problem. Using \emph{abstract Lagrangian duality} and introducing a penalty function, we propose a technique to tackle the integer variables issue. Finally, we obtain a convex optimization problem with continuous feasible set which can be solved by using tools for solving convex problems such as CVX \cite{Grant}, \cite{Boyd}.
The details of our proposed method for addressing the optimization problem stated in (\ref{problem}) {will be} explained in what follows.

\subsection{{Tackling the Coupled Variables in Problem (\ref{problem})}}
{As mentioned before, the first step to address problem (\ref{problem}) is to tackle the coupled UL and DL power control and sub-carrier assignment variables in constraints (\ref{transmit power constraint of users})-(\ref{target QoS constraint for each user at DL}) as well as the objective function. To do this, we first define two new auxiliary power variables $\widetilde{p}^\textnormal{u}_{n,k}$ and $\widetilde{p}^\textnormal{d}_{n,k}$ as
\begin{equation}
\label{auxiliary variable}
\widetilde{p}^\textnormal{u}_{n,k}= p^\textnormal{u}_{n,k} x_{n,k}; \quad \widetilde{p}^\textnormal{d}_{n,k}= p^\textnormal{d}_{n,k} x_{n,k} \quad \forall n \in \mathcal{N},k \in \mathcal{K}.
\end{equation}
Then, replacing the original power variables $\bold{p}$ by auxiliary power variables $\displaystyle{\widetilde{\bold{p}}= [\widetilde{\bold{p}}^\textnormal{u},\widetilde{\bold{p}}^\textnormal{d}]^T}$, where $\widetilde{\bold{p}}^\textnormal{u}=[\widetilde{p}^\textnormal{u}_{n,k}]^T_{N\times K}$ and $\widetilde{\bold{p}}^\textnormal{d}=[\widetilde{p}^\textnormal{d}_{n,k}]^T_{N\times K}$, in constraints (\ref{transmit power constraint of users}) and (\ref{transmit power constraint of base stations}), they are reformulated as
\begin{equation}
\qquad \qquad \sum_{\forall k \in \mathcal{K}}{\widetilde{p}^\textnormal{u}_{n,k} } \leq \overline{P}_n, \qquad \qquad \forall n \in \mathcal{N}, \tag{8-2-1} \label{relaxed transmit power constraint of users 1} 
\end{equation}
and
\begin{equation}
\!\!\!\!\!\!\!\!\!\!\!\!\!\!\!\!\!\!\!\!\!\!\!\! \sum_{\forall n \in \mathcal{N}}{ \sum_{\forall k \in \mathcal{K}}{\widetilde{p}^\textnormal{d}_{n,k} } } \leq \overline{P}_\textnormal{BS}, \tag{8-3-1} \label{relaxed transmit power constraint of base station 1} 
\end{equation}
respectively. Now, we add two following constraints:
\begin{equation}
\widetilde{p}^\textnormal{u}_{n,k} \leq x_{n,k} \overline{P}_n, \qquad \qquad \forall n \in \mathcal{N}, k \in \mathcal{K}, \tag{8-2-2} \label{relaxed transmit power constraint of users 2} 
\end{equation}
and
\begin{equation}
\widetilde{p}^\textnormal{d}_{n,k} \leq x_{n,k} \overline{P}_\textnormal{BS}, \qquad \qquad \forall n \in \mathcal{N}, k \in \mathcal{K}. \tag{8-3-2} \label{relaxed transmit power constraint of base station 2} 
\end{equation}
The constraint (\ref{relaxed transmit power constraint of users 2}) ensures that if the sub-carrier $k$ is not assigned to the $n$th UE, i.e. $x_{n,k}=0$, the transmit power of the $n$th UE on the sub-carrier $k$ is also zero, i.e., $\widetilde{p}^\textnormal{u}_{n,k}=0$. Also, if $x_{n,k}=1$, then $\widetilde{p}^\textnormal{u}_{n,k}$ at most can reach $\overline{P}_n$. Similarly, the constraint (\ref{relaxed transmit power constraint of base station 2}) ensures that if the sub-carrier $k$ is not assigned to the BS for transmitting information to the $n$th UE, i.e. $x_{n,k}=0$, the transmit power of the BS on the sub-carrier $k$ is also zero, i.e., $\widetilde{p}^\textnormal{d}_{n,k}=0$. Also, if $x_{n,k}=1$, then $\widetilde{p}^\textnormal{d}_{n,k}$ at most can reach $\overline{P}_\textnormal{BS}$.
In addition, the constraint (\ref{non-negative transmit power}) is rewritten as
\begin{equation}
\label{relaxed non-negative transmit power}
\widetilde{p}^\textnormal{u}_{n,k} , \widetilde{p}^\textnormal{d}_{n,k} \geq 0, \qquad \! \qquad  \forall n \in \mathcal{N},k \in \mathcal{K}. \tag{8-6-1} 
\end{equation}}

{Now, we rewrite constraints (\ref{target QoS constraint for each user at UL}) and (\ref{target QoS constraint for each user at DL}) as
\begin{equation}
\label{target QoS constraint for each user at UL with auxiliary power}
\widetilde{R}^\textnormal{u}_n(\bold{x}, \widetilde{\bold{p}}) \geq \overline{R}^\textnormal{u}_n, \qquad \qquad \qquad \qquad \!\!\! \quad \!\!\!  \forall n \in \mathcal{N}, \tag{8-4-1}
\end{equation} 
and
\begin{equation}
\label{target QoS constraint for each user at DL with auxiliary power}
\widetilde{R}^\textnormal{d}_n(\bold{x}, \widetilde{\bold{p}}) \geq \overline{R}^\textnormal{d}_n, \qquad \qquad \qquad \qquad \!\!\! \quad \!\!\!  \forall n \in \mathcal{N}, \tag{8-5-1}
\end{equation} 
respectively, where $\displaystyle{\widetilde{R}^\textnormal{u}_n(\bold{x}, \widetilde{\bold{p}})=\sum_{\forall k \in \mathcal{K}}{ \log \left(1+\widetilde{\gamma}^\textnormal{u}_{n,k}(\bold{x}, \widetilde{\bold{p}})\right) } }$, in which $\displaystyle{\widetilde{\gamma}^\textnormal{u}_{n,k}(\bold{x}, \widetilde{\bold{p}}) = \frac{\widetilde{p}^\textnormal{u}_{n,k} h_{n,k}}{{\delta_\textnormal{BS} {|l^\textnormal{SI}_\textnormal{BS}|^2}} \widetilde{p}^\textnormal{d}_{n,k} + N_\textnormal{BS}} }$ obtained by replacing $p^\textnormal{u}_{n,k} x_{n,k}$ by  $\widetilde{p}^\textnormal{u}_{n,k}$ and $p^\textnormal{d}_{n,k} x_{n,k}$ by  $\widetilde{p}^\textnormal{d}_{n,k}$ in the UL SINR and the UL transmission rate function in (\ref{gamma-UL}) and (\ref{UL-rate}), respectively. Similarly, $\displaystyle{\widetilde{R}^\textnormal{d}_n(\bold{x}, \widetilde{\bold{p}})}$ is found.}

{Now, we tackle the coupled UL and DL power control and sub-carrier assignment variables in the objective function of problem (\ref{problem}). To do this, we rewrite the total system sum-rate as
\begin{equation}
\label{total rate with auxilirly variables}
\widetilde{R}(\bold{x}, \widetilde{\bold{p}})=\sum_{\forall n \in \mathcal{N}}{\Big(\widetilde{R}^\textnormal{u}_n(\bold{x}, \widetilde{\bold{p}}) + \widetilde{R}^\textnormal{d}_n(\bold{x}, \widetilde{\bold{p}})\Big)},
\end{equation}
where, $\widetilde{R}^\textnormal{u}_n(\bold{x}, \widetilde{\bold{p}})$ and $\widetilde{R}^\textnormal{d}_n(\bold{x}, \widetilde{\bold{p}})$ are described in (\ref{target QoS constraint for each user at UL with auxiliary power}) and (\ref{target QoS constraint for each user at DL with auxiliary power}), respectively.
Also, we rewrite the total system power consumption in (\ref{totalPowerConsumption}) as
\begin{equation}
\label{total power consumption with auxilirly variables}
\widetilde{P}^\textnormal{T}\!(\bold{x}, \widetilde{\bold{p}}) \!=\! P_\textnormal{BS}^\textnormal{C} +\!\! \sum_{\forall n \in \mathcal{N}}\!\!\!{P_n^\textnormal{C}} +\!\! \sum_{\forall n \in \mathcal{N}}\!{ \sum_{\forall k \in \mathcal{K}}\!\!{ \!\left(\! {\frac{1}{\epsilon_n}} \widetilde{p}^\textnormal{u}_{n,k} \!+\! {\frac{1}{\epsilon_\textnormal{BS}}} \widetilde{p}^\textnormal{d}_{n,k}\!\right) } }.
\end{equation}
Finally, we rewrite the EE criterion in (\ref{energyEfficiency}) as
\begin{equation}
\label{energyEfficiency relaxed}
\widetilde{EE}(\bold{x}, \widetilde{\bold{p}})=\frac{\widetilde{R}(\bold{x}, \widetilde{\bold{p}})}{\widetilde{P}^\textnormal{T}(\bold{x}, \widetilde{\bold{p}})}.
\end{equation}}

{After applying the mentioned steps, problem (\ref{problem}) is reformulated as
\begin{align}
\underset{\widetilde{\bold{p}},\bold{x}}{\mathrm{maximize}}& \quad \widetilde{EE}(\bold{x}, \widetilde{\bold{p}})  \label{problem decoupled}\\
\mathrm{subject\ to} & \quad \textnormal{(\ref{exclusive allocation of subcarriers}), (\ref{relaxed transmit power constraint of users 1}), (\ref{relaxed transmit power constraint of users 2}), (\ref{relaxed transmit power constraint of base station 1}), (\ref{relaxed transmit power constraint of base station 2})}, \nonumber \\
& \quad \textnormal{(\ref{target QoS constraint for each user at UL with auxiliary power}), (\ref{target QoS constraint for each user at DL with auxiliary power}), (\ref{relaxed non-negative transmit power}) and (\ref{binary subcarrier and base station assignment})}. \nonumber
\end{align}
The problem (\ref{problem decoupled}) is an optimization problem with decoupled variables. However, it is still non-convex due to the non-convexity of the fractional objective function and constraints (\ref{target QoS constraint for each user at UL with auxiliary power}) and (\ref{target QoS constraint for each user at DL with auxiliary power}) as well as the combinatorial constraint (\ref{binary subcarrier and base station assignment}) on the sub-carrier assignment variables. In the next subsections, we first describe a technique to treat the fractional objective function in problem (\ref{problem decoupled}), then we tackle the non-convexity of constraints (\ref{target QoS constraint for each user at UL with auxiliary power}) and (\ref{target QoS constraint for each user at DL with auxiliary power}) and finally, we propose a technique to handle the integer sub-carrier assignment variables in (\ref{binary subcarrier and base station assignment}).}

\subsection{Transforming the Fractional Objective Function in (\ref{problem decoupled})}
In this section, we tackle the fractional objective function of (\ref{problem decoupled}).
Suppose $\mathcal{F}$ is the set of feasible solutions to problem (\ref{problem decoupled}) spanned by constraints (\ref{exclusive allocation of subcarriers}), (\ref{relaxed transmit power constraint of users 1}), (\ref{relaxed transmit power constraint of users 2}), (\ref{relaxed transmit power constraint of base station 1}), (\ref{relaxed transmit power constraint of base station 2}), (\ref{target QoS constraint for each user at UL with auxiliary power}), (\ref{target QoS constraint for each user at DL with auxiliary power}), (\ref{relaxed non-negative transmit power}) and (\ref{binary subcarrier and base station assignment}). We denote $q^*$ as the optimal EE in problem (\ref{problem decoupled}) represented as
\begin{equation}
q^* = \frac{\widetilde{R}(\bold{x^*},\bold{\widetilde{p}^*})}{\widetilde{P}^\textnormal{T}(\bold{x^*},\bold{\widetilde{p}^*})} = \underset{\bold{\widetilde{p}},\bold{x}\in \mathcal{F}}{\mathrm{maximize}}{ \frac{\widetilde{R}(\bold{x},\bold{\widetilde{p}})}{\widetilde{P}^\textnormal{T}(\bold{x},\bold{\widetilde{p}})} },
\end{equation}
where $\bold{x^*}$ and $\bold{\widetilde{p}^*}$ are the optimal sub-carrier assignment and the optimal transmit power control {for} problem (\ref{problem decoupled}), respectively. Now, we employ the following Theorem borrowed from non-linear fractional programming \cite{Dinkelbach} to address (\ref{problem decoupled}).
\begin{theorem}
\label{fractional programming}
{\cite{Dinkelbach}} The resource allocation policy attains the optimal EE, i.e., $q^*$, if and only if
\begin{align}
\label{equivalent problem}
\underset{\bold{\widetilde{p}},\bold{x} \in \mathcal{F}}{\mathrm{maximize}} \quad { \widetilde{R}(\bold{x,\widetilde{p}})-q^*\widetilde{P}^\textnormal{T}(\bold{x,\widetilde{p}})} \nonumber \\
= \widetilde{R}(\bold{x^*,\widetilde{p}^*})-q^*\widetilde{P}^\textnormal{T}(\bold{x^*,\widetilde{p}^*}) =0 ,
\end{align}
for $\widetilde{R}(\bold{x,\widetilde{p}}) \geq 0$ and $\widetilde{P}^\textnormal{T}(\bold{x,\widetilde{p}}) \geq 0$, where $\bold{x^*}$ and $\bold{\widetilde{p}^*}$ yield the optimal solution {to} problem (\ref{problem decoupled}).
\end{theorem}
\vspace{-20pt}
\begin{proof}
The proof is directly obtained {from} \cite{Dinkelbach}. 
\end{proof}

{A necessary and sufficient condition to obtain the optimal resource allocation policy for problem (\ref{problem decoupled}) is described in Theorem \ref{fractional programming}. The above theorem states that for the fractional objective function {of} problem (\ref{problem decoupled}), there is a transformed objective function in subtractive form (i.e., $\underset{\bold{\widetilde{p}},\bold{x} \in \mathcal{F}}{\mathrm{maximize}} \quad  \widetilde{R}(\bold{x,\widetilde{p}})-q^*\widetilde{P}^\textnormal{T}(\bold{x,\widetilde{p}})$), which shares the same resource allocation policy. Therefore, we focus on this transformed objective function in the rest of this paper.}

\subsection{{Iterative Dinkelbach Algorithm to Address Problem (\ref{problem decoupled}) with Transformed Objective Function}}
In this section, we propose an iterative algorithm known as Dinkelbach method \cite{Dinkelbach} for addressing problem (\ref{problem decoupled}) with a transformed objective function. 
The proposed {scheme} is given in Algorithm \ref{Dinkelbach algorithm}.

\begin{algorithm}
\caption{Iterative Resource Allocation Algorithm (Dinkelbach Method)}
\label{Dinkelbach algorithm}
\centering
\begin{algorithmic}[1]
\REQUIRE{$i=0, q_0=0, \Delta>0$} \\
\STATE $\quad$ $i$ : Dinkelbach iteration index  \\
\STATE $\quad$ $q_i$ : Dinkelbach parameter \\
\STATE $\quad$ $\Delta$ : The maximum allowed tolerance for convergence of Dinkelbach Method \\
\WHILE {$q_i - q_{i-1} > \Delta$} 
\STATE {Solve problem (\ref{probi}) using Algorithm \ref{algorithm for solving probi} and obtain resource allocation policy $\{ \bold{\widetilde{p}}^*_i,\bold{x}^*_i \}$} \\
\STATE  Set $i=i+1$ \\
\STATE  Set {$q_i = \displaystyle{\frac{\widetilde{R}(\bold{x}^*_i,\bold{\widetilde{p}}^*_i)}{\widetilde{P}^\textnormal{T}(\bold{x}^*_i,\bold{\widetilde{p}}^*_i)}}$}\\
\ENDWHILE \\
\STATE Set {$\{\bold{\widetilde{p}^*,x^*}\} = \{ \bold{\widetilde{p}}^*_{i-1},\bold{x}^*_{i-1} \}$} \\
\RETURN $\quad  {\bold{\widetilde{p}^*,x^*}}$
\end{algorithmic}
\end{algorithm}

{In accordance with Algorithm \ref{Dinkelbach algorithm}, given parameter $q_i$ in iteration $i$ of Dinkelbach method, the following optimization problem should be solved:
\begin{align}
\underset{\widetilde{\bold{p}}, \bold{x}}{\mathrm{maximize}} & \quad \widetilde{R}(\bold{x,\widetilde{p}})-q_i \widetilde{P}^\textnormal{T}(\bold{x,\widetilde{p}}) \label{probi} \\
\mathrm{subject\ to} & \quad \textnormal{(\ref{exclusive allocation of subcarriers}), (\ref{relaxed transmit power constraint of users 1}), (\ref{relaxed transmit power constraint of users 2}), (\ref{relaxed transmit power constraint of base station 1}), (\ref{relaxed transmit power constraint of base station 2})}, \nonumber \\
& \quad \textnormal{(\ref{target QoS constraint for each user at UL with auxiliary power}), (\ref{target QoS constraint for each user at DL with auxiliary power}), (\ref{relaxed non-negative transmit power}) and (\ref{binary subcarrier and base station assignment})}. \nonumber
\end{align}
Then, $q_i$ is updated by {$q_i = \displaystyle{\frac{\widetilde{R}(\bold{x}^*_i,\bold{\widetilde{p}}^*_i)}{\widetilde{P}^\textnormal{T}(\bold{x}^*_i,\bold{\widetilde{p}}^*_i)}}$}, where {$\{\bold{\widetilde{p}}^*_i, \bold{x}^*_i\}$} is the resource allocation policy corresponding to problem (\ref{probi}). The algorithm is terminated when $q_i$ converges 
and so the solution to problem (\ref{problem decoupled}), i.e., $\{\bold{\widetilde{p}^*,x^*}\}$, is eventually achieved. 
The proposed algorithm converges to the optimal solution of problem (\ref{problem decoupled}), if we are able to solve the inner problem (\ref{probi}) in each iteration. The convergence proof is similar to the proof given in \cite{Dinkelbach}. However, as the inner problem (\ref{probi}) is a non-convex problem, we employ MM approach by constructing a sequence of surrogate functions using Taylor approximation to make it convex. So, the proposed algorithm converges to a locally optimal solution of the problem (\ref{problem decoupled}). Also, our simulation results show that our proposed algorithm closely achieves the globally optimal solution. In the next subsection, we derive the solution to the problem (\ref{probi}).}

\subsection{{Solving the Optimization Problem (\ref{probi})}}
In this section, we obtain our proposed resource allocation policy for addressing problem (\ref{probi}). To do this, we first make the problem (\ref{probi}) convex by applying MM algorithm and then we handle the integer sub-carrier assignment variables by using \emph{abstract Lagrangian duality}.

Although, we addressed the issue of the coupled variables in constraints (\ref{target QoS constraint for each user at UL with auxiliary power}) and (\ref{target QoS constraint for each user at DL with auxiliary power}), they are still non-convex due to the logarithmic rate function. In order to handle this issue, we employ MM algorithm \cite{Sun} to make constraints (\ref{target QoS constraint for each user at UL with auxiliary power}) and (\ref{target QoS constraint for each user at DL with auxiliary power}) convex as explained in what follows. First, by replacing $\widetilde{R}^\textnormal{u}_n(\bold{x}, \widetilde{\bold{p}})$ in (\ref{target QoS constraint for each user at UL with auxiliary power}), we have $\displaystyle{\sum_{\forall k \in \mathcal{K} }{\log \left(1+\frac{\widetilde{p}^\textnormal{u}_{n,k} h_{n,k}}{{\delta_\textnormal{BS} {|l^\textnormal{SI}_\textnormal{BS}|^2}} \widetilde{p}^\textnormal{d}_{n,k} + N_\textnormal{BS}}\right)} \geq \overline{R}^\textnormal{u}_n, \quad \forall n \in \mathcal{N}}$. Rearranging this relation, we rewrite it as
\begin{align}
\label{Rearranged 10} 
&  \sum_{\forall k \in \mathcal{K} }{\log_{}{\left({\delta_\textnormal{BS} {|l^\textnormal{SI}_\textnormal{BS}|^2}} \widetilde{p}^\textnormal{d}_{n,k} + N_\textnormal{BS}+\widetilde{p}^\textnormal{u}_{n,k} h_{n,k} \right) }} \nonumber \\
& -\sum_{\forall k \in \mathcal{K} }{\log_{}{\left( {\delta_\textnormal{BS} {|l^\textnormal{SI}_\textnormal{BS}|^2}} \widetilde{p}^\textnormal{d}_{n,k} + N_\textnormal{BS} \right)}} \geq \overline{R}^\textnormal{u}_n, \quad \forall n \in \mathcal{N}.
\end{align}
We redefine the left side of the above equation as a difference of convex functions (DC) as
\begin{equation}
\label{UL rate constraint in dc form}
f^\textnormal{u}_1(\bold{x}, \widetilde{\bold{p}}) - f^\textnormal{u}_2(\bold{x}, \widetilde{\bold{p}}) \geq \overline{R}^\textnormal{u}_n, \quad \forall n \in \mathcal{N},
\end{equation}
where, $\displaystyle{f^\textnormal{u}_1(\bold{x}, \widetilde{\bold{p}})\!=\!\!\sum_{\forall k \in \mathcal{K} }{\log_{}{\left({\delta_\textnormal{BS} {|l^\textnormal{SI}_\textnormal{BS}|^2}} \widetilde{p}^\textnormal{d}_{n,k} \!+\! N_\textnormal{BS}+\widetilde{p}^\textnormal{u}_{n,k} h_{n,k} \right) }}}$ and $\displaystyle{f^\textnormal{u}_2(\bold{x}, \widetilde{\bold{p}})=\sum_{\forall k \in \mathcal{K} }{\log_{}{\left( {\delta_\textnormal{BS} {|l^\textnormal{SI}_\textnormal{BS}|^2}} \widetilde{p}^\textnormal{d}_{n,k} + N_\textnormal{BS} \right)}}}$. 
Even though both $f^\textnormal{u}_1(\bold{x}, \widetilde{\bold{p}})$ and $f^\textnormal{u}_2(\bold{x}, \widetilde{\bold{p}})$ are concave, the subtraction of two concave functions is not necessarily concave \cite{Sun}. To obtain a concave approximation for the constraint (\ref{UL rate constraint in dc form}), we apply MM algorithm \cite{Sun} and construct a surrogate function for $f^\textnormal{u}_2(\bold{x}, \widetilde{\bold{p}})$ using first order Taylor approximation as
\begin{equation}
\label{taylor approximation for UL rate function}
\widetilde{f}^\textnormal{u}_2(\bold{x}, \widetilde{\bold{p}})=f^\textnormal{u}_2(\bold{x}, \widetilde{\bold{p}}^{(t-1)}) + \nabla_{\widetilde{\bold{p}}}{{f^\textnormal{u}_2}^T(\bold{x}, \widetilde{\bold{p}}^{(t-1)})} (\widetilde{\bold{p}}-\widetilde{\bold{p}}^{(t-1)}),
\end{equation}
where, $\widetilde{\bold{p}}^{(t-1)}$ is the solution of the problem at $(t-1)$th iteration,  and $\nabla_{\widetilde{\bold{p}}}$ is the gradient operation with respect to $\widetilde{\bold{p}}$. Approximation (\ref{taylor approximation for UL rate function}) satisfies the MM principles and makes a tight lower bound of $\widetilde{R}^\textnormal{u}_n(\bold{x}, \widetilde{\bold{p}})$ \cite{Sun}. Now, we have the following constraint which is convex:
\begin{equation}
\label{convex UL rate constraint}
f^\textnormal{u}_1(\bold{x}, \widetilde{\bold{p}}) - \widetilde{f}^\textnormal{u}_2(\bold{x}, \widetilde{\bold{p}}) \geq \overline{R}^\textnormal{u}_n, \quad \forall n \in \mathcal{N}. \tag{8-4-2}
\end{equation}
Similarly, constraint (\ref{target QoS constraint for each user at DL with auxiliary power}) is rewritten as the following convex constraint:
\begin{equation}
\label{convex DL rate constraint}
f^\textnormal{d}_1(\bold{x}, \widetilde{\bold{p}}) - \widetilde{f}^\textnormal{d}_2(\bold{x}, \widetilde{\bold{p}}) \geq \overline{R}^\textnormal{d}_n, \quad \forall n \in \mathcal{N}, \tag{8-5-2}
\end{equation}
where, $f^\textnormal{d}_1(\bold{x}, \widetilde{\bold{p}})$ and $\widetilde{f}^\textnormal{d}_2(\bold{x}, \widetilde{\bold{p}})$ are obtained in a similar way of obtaining $f^\textnormal{u}_1(\bold{x}, \widetilde{\bold{p}})$ and $\widetilde{f}^\textnormal{u}_2(\bold{x}, \widetilde{\bold{p}})$ as explained.

{Now, we treat the non-convex total system sum-rate function, i.e. $\widetilde{R}(\bold{x}, \widetilde{\bold{p}})$ in the objective function of problem (\ref{probi}). Replacing $\widetilde{R}^\textnormal{u}_n(\bold{x}, \widetilde{\bold{p}})$ and $\widetilde{R}^\textnormal{d}_n(\bold{x}, \widetilde{\bold{p}})$ from (\ref{target QoS constraint for each user at UL with auxiliary power}) and (\ref{target QoS constraint for each user at DL with auxiliary power}), respectively, we rewrite (\ref{total rate with auxilirly variables}) as  
\begin{align}
\label{total rate replaced}
& \widetilde{R}(\bold{x}, \widetilde{\bold{p}})=\sum_{\forall n \in \mathcal{N}}{\sum_{\forall k \in \mathcal{K} }}{\Bigg(\log \left(1+\frac{\widetilde{p}^\textnormal{u}_{n,k} h_{n,k}}{{\delta_\textnormal{BS} {|l^\textnormal{SI}_\textnormal{BS}|^2}} \widetilde{p}^\textnormal{d}_{n,k} + N_\textnormal{BS}}\right)}\nonumber \\
& \qquad \quad + \log \left(1+\frac{\widetilde{p}^\textnormal{d}_{n,k} g_{n,k}}{{\delta_n {|l^\textnormal{SI}_n|^2}} \widetilde{p}^\textnormal{u}_{n,k} + N_n}\right)\Bigg).
\end{align}
Rearranging (\ref{total rate replaced}), we rewrite it as
\begin{align}
\label{total rate rearranged}
& \widetilde{R}(\bold{x}, \widetilde{\bold{p}})\!=\!\!\!\sum_{\forall n \in \mathcal{N}}{\sum_{\forall k \in \mathcal{K} }}{\Bigg(\log \left({\delta_\textnormal{BS} {|l^\textnormal{SI}_\textnormal{BS}|^2}} \widetilde{p}^\textnormal{d}_{n,k} \!+\! N_\textnormal{BS}+\widetilde{p}^\textnormal{u}_{n,k} h_{n,k} \right)}  \nonumber \\
& \quad + \log \left({\delta_n {|l^\textnormal{SI}_n|^2}} \widetilde{p}^\textnormal{u}_{n,k} + N_n + \widetilde{p}^\textnormal{d}_{n,k} g_{n,k} \right) \nonumber \\
& \quad \!-\! \log \!\left({\delta_\textnormal{BS} {|l^\textnormal{SI}_\textnormal{BS}|^2}} \widetilde{p}^\textnormal{d}_{n,k} \!+\!\! N_\textnormal{BS} \right) \!-\! \log \!\left({\delta_n {|l^\textnormal{SI}_n|^2}} \widetilde{p}^\textnormal{u}_{n,k} \!+\!\! N_n \right)\!\!\!\Bigg).
\end{align} 
We redefine the above equation as a DC as
\begin{equation}
\label{total rate in dc form}
\widetilde{R}(\bold{x}, \widetilde{\bold{p}})= f(\bold{x}, \widetilde{\bold{p}})- g(\bold{x}, \widetilde{\bold{p}}),
\end{equation}
where, $\displaystyle{f(\bold{x}, \widetilde{\bold{p}})\!=\!\!\!\!\sum_{\forall n \in \mathcal{N}}{\sum_{\forall k \in \mathcal{K} }}\!\!\Big(\!\log ({\delta_\textnormal{BS} {|l^\textnormal{SI}_\textnormal{BS}|^2}} \widetilde{p}^\textnormal{d}_{n,k} \!\!+\!\!  N_\textnormal{BS}\!\!+\!\!\widetilde{p}^\textnormal{u}_{n,k} h_{n,k} )}$ $\displaystyle{\!+\!\log \left({\delta_n {|l^\textnormal{SI}_n|^2}} \widetilde{p}^\textnormal{u}_{n,k} \!+\! N_n + \widetilde{p}^\textnormal{d}_{n,k} g_{n,k} \right)\Big)}$, and $\displaystyle{g(\bold{x}, \widetilde{\bold{p}})=} $ $\displaystyle{\sum_{\forall n \in \mathcal{N}}\!{\sum_{\forall k \in \mathcal{K} }} \!\!\Big(\!\log ({\delta_\textnormal{BS} {|l^\textnormal{SI}_\textnormal{BS}|^2}} \widetilde{p}^\textnormal{d}_{n,k} \!\!+\!\! N_\textnormal{BS} )\!+\!\log \left({\delta_n {|l^\textnormal{SI}_n|^2}} \widetilde{p}^\textnormal{u}_{n,k} \!\!+\!\! N_n \right)\Big)}$. 
After applying the mentioned steps, the problem (\ref{probi}) is reformulated as\footnote{{The objective function of problem (\ref{probi relaxed}) is still non-concave which later becomes concave using MM algorithm in (\ref{taylor approximation}).}}
\begin{align}
\underset{\widetilde{\bold{p}}, \bold{x}}{\mathrm{maximize}} & \quad f(\bold{x}, \widetilde{\bold{p}})- g(\bold{x}, \widetilde{\bold{p}})-q_i \widetilde{P}^\textnormal{T}(\bold{x,\widetilde{p}}) \label{probi relaxed} \\
\mathrm{subject\ to} & \quad \textnormal{(\ref{exclusive allocation of subcarriers}), (\ref{relaxed transmit power constraint of users 1}), (\ref{relaxed transmit power constraint of users 2}), (\ref{relaxed transmit power constraint of base station 1}), (\ref{relaxed transmit power constraint of base station 2})}, \nonumber \\
& \quad \textnormal{(\ref{convex UL rate constraint}), (\ref{convex DL rate constraint}), (\ref{relaxed non-negative transmit power}) and (\ref{binary subcarrier and base station assignment})}. \nonumber
\end{align}}

Next, we handle the issue of incorporating integer variables on the objective function as well as constraints. To do this, we first relax the sub-carrier assignment variable $x_{n,k}$ to be a real value between zero and one alternative to a binary value. So, the constraint (\ref{binary subcarrier and base station assignment}) is rewritten as
\begin{equation}
\label{relaxed binary subcarrier assignment}
x_{n,k} \in [0,1], \qquad \quad \forall n \in \mathcal{N}, k \in \mathcal{K}. \tag{8-7-1} \nonumber
\end{equation}
Now, inspired by the approach in \cite{Che}, we force the relaxed sub-carrier assignment variables to take binary values by defining a new constraint as 
\begin{equation}
\label{new subcarrier assignment constraint}
\sum_{\forall n \in \mathcal{N}}{\sum_{\forall k \in \mathcal{K}}{\left( x_{n,k} - (x_{n,k})^2 \right)}} \leq 0. \nonumber \tag{8-8}
\end{equation}
This constraint is satisfied for binary values, i.e., $x_{n,k} \in \{0,1\}$. Adding the new constraint (\ref{new subcarrier assignment constraint}) to the problem (\ref{probi relaxed}), the new optimization problem is:
\begin{align}
\underset{\widetilde{\bold{p}}, \bold{x}}{\mathrm{maximize}} & \quad f(\bold{x}, \widetilde{\bold{p}})- g(\bold{x}, \widetilde{\bold{p}})-q_i \widetilde{P}^\textnormal{T}(\bold{x,\widetilde{p}}) \label{new problem} \\
\mathrm{subject\ to} & \quad {\textnormal{(\ref{exclusive allocation of subcarriers}), (\ref{relaxed transmit power constraint of users 1}), (\ref{relaxed transmit power constraint of users 2}), (\ref{relaxed transmit power constraint of base station 1}), (\ref{relaxed transmit power constraint of base station 2})}}, \nonumber \\
& \quad {\textnormal{(\ref{convex UL rate constraint}), (\ref{convex DL rate constraint}), (\ref{relaxed non-negative transmit power}), (\ref{relaxed binary subcarrier assignment}) and (\ref{new subcarrier assignment constraint})}}. \nonumber
\end{align}
However, the constraint (\ref{new subcarrier assignment constraint}) is not convex. In order to treat the non-convexity of (\ref{new subcarrier assignment constraint}), using the \emph{abstract Lagrangian duality}, we add the constraint (\ref{new subcarrier assignment constraint}) as a penalty term to the objective function of problem (\ref{new problem}). More specifically, the abstract Lagrangian function of problem (\ref{new problem}) with only one Lagrangian multiplier to handle the non-convex constraint (\ref{new subcarrier assignment constraint}) is given by 
\begin{align}
\label{lagrangian function}
\bold{L}(\mathbf{x},\widetilde{\bold{p}})\!=\! f(\bold{x}, \widetilde{\bold{p}})- g(\bold{x}, \widetilde{\bold{p}})-q_i \widetilde{P}^\textnormal{T}(\bold{x,\widetilde{p}}) \nonumber \\
- \lambda \sum_{\forall n \in \mathcal{N}}\!{\sum_{\forall k \in \mathcal{K}}{\!\!\left( x_{n,k} - (x_{n,k})^2 \right)}},
\end{align}
where $\lambda$ acts as a penalty factor to penalize the objective function when the value of $x_{n,k}$ is not binary. Letting the feasible set spanned by constraints (\ref{exclusive allocation of subcarriers}), (\ref{relaxed transmit power constraint of users 1}), (\ref{relaxed transmit power constraint of users 2}), (\ref{relaxed transmit power constraint of base station 1}), (\ref{relaxed transmit power constraint of base station 2}), (\ref{convex UL rate constraint}), (\ref{convex DL rate constraint}), (\ref{relaxed non-negative transmit power}) and (\ref{relaxed binary subcarrier assignment}) be denoted by $\mathcal{D}$, the problem (\ref{new problem}) is expressed by $\displaystyle{\max_{(\widetilde{\bold{p}},\bold{x}) \in \mathcal{D}}{\min_{\lambda\geq0}{\bold{L}(\mathbf{x},\widetilde{\bold{p}})}}}$, while its dual problem is $\displaystyle{\min_{\lambda\geq0}{\max_{(\widetilde{\bold{p}},\bold{x}) \in \mathcal{D}}{\bold{L}(\mathbf{x},\widetilde{\bold{p}})}}}$. Note that, in general, there is a duality gap, i.e., $\displaystyle{\max_{(\widetilde{\bold{p}},\bold{x}) \in \mathcal{D}}{\min_{\lambda\geq0}{\bold{L}(\mathbf{x},\widetilde{\bold{p}})}} \leq \min_{\lambda\geq0}{\max_{(\widetilde{\bold{p}},\bold{x}) \in \mathcal{D}}{\bold{L}(\mathbf{x},\widetilde{\bold{p}})}}}$ \cite{Ata}, \cite{Che}. 

{\begin{proposition}
\label{proposition1} 
For adequately large values of penalty factor, i.e., $\lambda$, the optimization problem (\ref{new problem}) is equivalent to the following problem:
\begin{align}
\underset{\widetilde{\bold{p}},\bold{x}}{\mathrm{maximize}}& \quad \bold{L}(\mathbf{x},\widetilde{\bold{p}}) \label{new problem with penalty term} \\
\mathrm{subject\ to} & \quad \widetilde{\bold{p}},\bold{x} \in \mathcal{D}. \nonumber
\end{align}
\end{proposition}
\begin{proof}
The proof is given in the Appendix \ref{appendix A}.
\end{proof}}

{However, the objective function of problem (\ref{new problem with penalty term}) is still non-concave. In order to handle this issue, we apply MM algorithm. To do this, we first redefine the objective function of (\ref{new problem with penalty term}) as 
\begin{equation}
\label{objective function in dc form}
\bold{L}(\mathbf{x},\widetilde{\bold{p}})=e_1(\bold{x},\widetilde{\bold{p}}) - e_2(\bold{x},\widetilde{\bold{p}}),
\end{equation}
where, $\displaystyle{e_1(\bold{x},\widetilde{\bold{p}})\!=\!f(\bold{x,\widetilde{p}})-q_i \widetilde{P}^\textnormal{T}(\bold{x,\widetilde{p}}) \!-\! \lambda \sum_{\forall n \in \mathcal{N}}{\sum_{\forall k \in \mathcal{K}}{x_{n,k}}}}$ and $\displaystyle{e_2(\bold{x},\widetilde{\bold{p}})=g(\bold{x,\widetilde{p}})-\lambda \sum_{\forall n \in \mathcal{N}}{\sum_{\forall k \in \mathcal{K}}{ (x_{n,k})^2}}}$. In fact, although both $e_1(\bold{x},\widetilde{\bold{p}})$ and $e_2(\bold{x}, \widetilde{\bold{p}})$ are concave, the subtraction of them is not necessarily concave \cite{Sun}. It is obvious that the objective function of (\ref{objective function in dc form}) belongs to the class of DC programming. To obtain a concave approximation for the objective function in (\ref{objective function in dc form}), we apply MM approach \cite{Sun} and construct a sequence of surrogate functions for $e_2(\bold{x}, \widetilde{\bold{p}})$ using first order Taylor approximation as}
{
\begin{align}
\label{taylor approximation}
\widetilde{e}_2\big(\bold{x}\!,\!\widetilde{\bold{p}}\big)\!\!=&e_2\big(\bold{x}^{(t-1)}\!,\!\widetilde{\bold{p}}^{(t-1)}\big)\!\!+\!\! \nabla_{\widetilde{\bold{p}}}{e_2^T\!\big(\bold{x}^{(t-1)}\!,\!\widetilde{\bold{p}}^{(t-1)}\big)} \!\big(\widetilde{\bold{p}}\!-\!\widetilde{\bold{p}}^{(t-1)}\big)\nonumber\\
&+\nabla_{\bold{x}}e_2^T\big(\bold{x}^{(t-1)}, \widetilde{\bold{p}}^{(t-1)}\big)\big(\bold{x}-\bold{x}^{(t-1)}\big),
\end{align}
where, $\widetilde{\bold{p}}^{(t-1)}$ and $\bold{x}^{(t-1)}$ are the solution of the problem at $(t-1)$th iteration and $\nabla_{\widetilde{\bold{p}}}$  and $\nabla_{\bold{x}}$ present the gradient operation with respect to $\widetilde{\bold{p}}$ and $\bold{x}$,~respectively. Since $e_2(\bold{x}, \widetilde{\bold{p}})$ is a concave function, due to the first order condition \cite{Boyd}, we have $\widetilde{e}_2(\bold{x}, \widetilde{\bold{p}}) \geq e_2(\bold{x}, \widetilde{\bold{p}})$ which shows that $\widetilde{e}_2(\bold{x}, \widetilde{\bold{p}})$ makes a tight lower bound of $\bold{L}(\bold{x},\widetilde{\bold{p}})$  \cite{Sun} as summarized in the following Proposition.}

{\begin{proposition}
\label{proposition2}
The approximation (\ref{taylor approximation}) satisfies MM principles and makes a tight lower bound of $\bold{L}(\mathbf{x},\widetilde{\bold{p}})$ which results in a sequence of improved solutions for problem (\ref{new problem with penalty term}) and yields a locally optimal solution\footnote{Achieving the globally optimal solution is not guaranteed.}. 
\end{proposition}
\begin{proof}
The proof is given in the Appendix \ref{appendix B}.
\end{proof}}

{Now, the objective function $\bold{L}(\mathbf{x},\widetilde{\bold{p}})=e_1(\bold{x}, \widetilde{\bold{p}}) - \widetilde{e}_2(\bold{x}, \widetilde{\bold{p}})$ is a concave function at each iteration. Finally, we obtain the following problem:
\begin{align}
\label{convex problem}
\underset{\widetilde{\bold{p}},\bold{x}}{\mathrm{maximize}}& \quad e_1(\bold{x}, \widetilde{\bold{p}}) - \widetilde{e}_2(\bold{x}, \widetilde{\bold{p}}) \\
\mathrm{subject\ to} & \quad \widetilde{\bold{p}},\bold{x} \in \mathcal{D}. \nonumber
\end{align}}

The problem in (\ref{convex problem}) is a convex optimization problem at each iteration,~it can be solved efficiently using the optimization package including interior point method such as CVX \cite{Grant}, \cite{Boyd}.
Therefore, we apply an iterative algorithm to tighten the
obtained lower bound where the solution of (\ref{convex problem}) in iteration $(t)$ is used as the initial point for the next iteration $(t+1)$. This iterative algorithm continues until reaches to a locally optimum point\footnote{{The proof is given in the Appendix \ref{appendix B}.}} of problem (\ref{new problem with penalty term}) or equivalently problem (\ref{new problem}) in a polynomial time complexity \cite{Reverse}, \cite{Ata}. The detailed scheme is provided in Algorithm \ref{algorithm for solving probi}.

\begin{algorithm}
\caption{{Iterative Algorithm to Solve Inner Problem (MM Method)}}
\label{algorithm for solving probi}
\centering
\begin{algorithmic}[1]
\REQUIRE{{$i, q_i, t=0, T_{\max}, \lambda, \widetilde{\bf{p}}^{(0)}, {\bf{x}}^{(0)}$}} \\
\STATE {$\quad$ $i$ : Dinkelbach iteration index}  \\
\STATE {$\quad$ $q_i$ : Dinkelbach parameter}  \\
\STATE {$\quad$ $t$ : MM iteration index}  \\
\STATE {$\quad$ $T_{\max}$ : Maximum number of iterations}  \\
\STATE {$\quad$ $\lambda>>1$ : Penalty factor} \\
\STATE {$\quad$ $\widetilde{\bf{p}}^{(0)}, {\bf{x}}^{(0)}$ : Initial value of $\widetilde{\bf{p}}$ and $\bf{x}$ in iteration 0} \\
\REPEAT 
\STATE {Calculate $f^\textnormal{u}_2(\bold{x}, \widetilde{\bold{p}}^{(t)})$, $\nabla_{\widetilde{\bold{p}}}{{f^\textnormal{u}_2}^T(\bold{x}, \widetilde{\bold{p}}^{(t)})}$, $f^\textnormal{d}_2(\bold{x}, \widetilde{\bold{p}}^{(t)})$, $\nabla_{\widetilde{\bold{p}}}{{f^\textnormal{d}_2}^T(\bold{x}, \widetilde{\bold{p}}^{(t)})}$, $e_2(\bold{x}^{(t)}, \widetilde{\bold{p}}^{(t)})$, $\nabla_{\widetilde{\bold{p}}}{e_2^T(\bold{x}^{(t)}, \widetilde{\bold{p}}^{(t)})}$ and $\nabla_{\bold{x}}{e_2^T(\bold{x}^{(t)}, \widetilde{\bold{p}}^{(t)})}$}
\STATE  {Set $t=t+1$} \\
\STATE {Solve problem (\ref{convex problem}) using CVX and obtain resource allocation policy $\{ \bold{\widetilde{p}}^{*(t)},\bold{x}^{*(t)} \}$} \\
\UNTIL {Convergence or $t=T_{\max}$} \\
\STATE {Set $\{\bold{\widetilde{p}}^*_i,\bold{x}^*_i\} = \{ \bold{\widetilde{p}}^{*(t)},\bold{x}^{*(t)} \}$} \\
\RETURN {$\quad  {\bold{\widetilde{p}}^*_i,\bold{x}^*_i}$}
\end{algorithmic}
\end{algorithm}

\subsubsection{Feasibility and Initial Feasible Allocation}\hspace*{\fill} \\ 
The optimization problem (\ref{convex problem}) is feasible since there exists an initial point, i.e., $\widetilde{\bf{p}}^{(0)}, {\bf{x}}^{(0)}$, which holds all of the constraints in $\mathcal{D}$. However, the iterative algorithm \ref{algorithm for solving probi} needs to find an appropriate initialization vector for the sub-carrier assignment and power control. The difficulty lies in how the initial point meets the QoS constraints (\ref{convex UL rate constraint}) and (\ref{convex DL rate constraint}). To meet these constraints, we assume that the UEs and BS have complete SIC capability, and so, there exists only noise in the system. 
For the considered system, we obtain the initial point ${\bf{x}}^{(0)}$ and $\widetilde{\bf{p}}^{(0)}$ in two steps. In step one, we assume that each sub-carrier is assigned to a UE with highest channel gain which gives the initial sub-carrier assignment, i.e., ${\bf{x}}^{(0)}$. In step two, for the initial sub-carrier assignment policy ${\bf{x}}^{(0)}$ obtained in step one, we find the initial power control policy, i.e., $\widetilde{\bf{p}}^{(0)}$, obtained by solving the classic water-filling problem using CVX.

\subsubsection{Computational Complexity Analysis}\hspace*{\fill} \\ 
The computational complexity of the proposed scheme in each iteration of the Dinkelbach method in Algorithm \ref{Dinkelbach algorithm} is dominated by the MM algorithm proposed in Algorithm \ref{algorithm for solving probi}. Since the optimization problem (\ref{convex problem}) consists of $NK$ variables and $3NK+3N+K+1$ linear convex constraints, its time complexity is given by $(NK)^3(3NK+3N+K+1)$ (asymptotically $ \approx (NK)^{4}$) which is polynomial time complexity \cite{Che}. Therefore, the overall complexity of our proposed scheme is of order $\mathcal{O}(I_\textnormal{Dinklebach}I_\textnormal{MM}(NK)^{4})$, where $I_\textnormal{Dinkelbach}$ and $I_\textnormal{MM}$ are the number of iterations required for reaching convergence in Dinkelbach and MM method, respectively. More specifically, $I_\textnormal{MM}$ is the required iterations for the D.C. programming with the interior point method employed by CVX to solve problem (\ref{convex problem}) given by $I_\textnormal{MM}=\log{\frac{3NK+3N+K+1}{t^{(0)} \Lambda \xi}}$, where $t^{(0)}$ is the initial point, $0<\!\Lambda\ll \!1$ is the stopping criterion, and $\xi$ is used for updating the accuracy of the method \cite{Boyd}, \cite{Khalili2}.

In this section, we proposed a novel resource allocation scheme for OFDMA networks with IBFD capability in both the BS and UEs. We maximized the EE while considering data rate requirements of UEs in DL and UL. Note that the proposed scheme can be applied to the multi-cell systems in which the UL transmission of the UEs and DL transmission of the BSs in a cell cause the inter-cell interference to the UEs and BSs in neighboring cells. In IBFD cellular networks with two-node architecture, the cell-edge UEs associated with different BSs can give severe inter-cell interference to each other \cite{nam}. To address this issue, simultaneous interference management in the BSs and UEs by employing the sub-carrier assignment, power control and user association in the resource allocation policy is required.

\section{Numerical Results}
In this section, we evaluate our proposed resource allocation scheme and compare it with existing schemes proposed in \cite{He}, \cite{wen}, \cite{nam}, and \cite{Leng} {to} demonstrate the efficacy of our scheme. We consider a single-cell system where the BS is placed at the midpoint of a square cell and the UEs are generated randomly inside the cell. {The number of UEs and sub-carriers are set as $N=10$ and $K=16$, respectively.} {The channel gains are independent and generated by {applying} Rayleigh fading, Log-Normal shadowing with standard deviation of $8$ dB, illustrated by an exponentially distributed random variable {of unit mean}, and the path loss model $\textnormal{PL}(d)=\textnormal{PL}_0 + 10 \theta \log_{10}{d}$. {Here}, $d$ is the distance between the corresponding UE and the BS, $\theta$ is the path loss exponent, and $\textnormal{PL}_0$ is the constant path loss coefficient which depends on the average channel attenuation and antenna characteristics \cite{Tang}. 
For the power consumption model, we set the constant circuit power consumed by {the} BS and UEs as $P_\textnormal{BS}^\textnormal{C} = 30$ dBm and $P_n^\textnormal{C} = 20$ dBm, respectively (as in \cite{loodaricheh} and \cite{Ng}). 
The power amplifier efficiency of {the} BS and UEs are set as $\epsilon_\textnormal{BS} = 30\%$ and $\epsilon_n = 20\%$ (as in \cite{powerModel} and \cite{Tang}), respectively.  
We set the maximum transmit power of {the} BS and UEs as $\overline{P}_\textnormal{BS} = 42$ dBm and $\overline{P}_n = 23$ dBm, similar to \cite{Ng}, \cite{crick} and \cite{Tang}. 
Parameters of {the} SIC constant for {the} BS and UEs are set as $\delta_\textnormal{BS} = -100$ dB and $\delta_n = -70$ dB, as in \cite{SIC2}, \cite{SICforBS}, \cite{SICforBS2} and \cite{SICforUE}. The fading coefficients of the SI channel at the BS and UEs (i.e., $l^\textnormal{SI}_\textnormal{BS}$ and $l^\textnormal{SI}_n$) are generated as independent and identically distributed Rician random variables with Rician factor $5$ dB, similar to \cite{Reverse} and \cite{Duarte}. 
Furthermore, the penalty factor to handle the integer sub-carrier assignment variables in our proposed scheme is set as $\lambda=10^{\log (\frac{\overline{P}_\textnormal{BS}}{N_\textnormal{BS}})}$.
The minimum data rate requirement in UL and DL for the $n$th UE are set as $\overline{R}^\textnormal{u}_n=\overline{R}^\textnormal{d}_n=2$ bps$/$Hz.
Other parameters such as the noise power and the sub-carrier bandwidth are set as $-120$ dBm and $180$ KHz, similar to \cite{Dong}, \cite{crick}, \cite{nam} and \cite{lohani}.}
The default values of all parameters used in the simulations are summarized in Table \ref{Simulation Parameters}. Unless mentioned otherwise, the default values are used. In all scenarios, the numerical results are obtained by averaging over 100 independent snapshots with randomly generated UEs' locations.
\begin{table}
\caption{simulation parameters}
\label{Simulation Parameters}
\centering
\begin{tabular}{|c|c|}\hline
{\bf Parameter} & {\bf Value} \\ \hline \hline
{Cell diameter} & {$250$ m} \\
{Number of UEs ($N$)} & {$10$} \\
{Number of sub-carriers ($K$)} & {$16$} \\
{Noise power at the BS and UEs ($N_\textnormal{BS}, N_n $)} & {$-120$} dBm \\
Sub-carrier bandwidth & {$180$} KHz \\
{Path loss exponent ($\theta$)} & {$3.76$} \\
{Constant path loss coefficient ($\textnormal{PL}_0$)} & {$128.1$ dB} \\
{SIC constant for {the} BS ($\delta_\textnormal{BS}$)} & {$-100$ dB} \\
{SIC constant for {the} UEs ($\delta_n$)} & {$-70$ dB} \\
{Fading coefficient of the SI channel} & {Rician factor $5$ dB} \\
{Power amplifier efficiency of {the} BS ($\epsilon_\textnormal{BS}$)}  & {$30\%$} \\
{Power amplifier efficiency of {the} $n$th UE ($\epsilon_n$)}  & {$20\%$} \\
Constant power consumed by the BS ($P_\textnormal{BS}^\textnormal{C}$) & {$30$ dBm} \\
Constant power consumed by the $n$th UE ($P_n^\textnormal{C}$) & {$20$ dBm} \\ 
Maximum transmit power of the BS ($\overline{P}_\textnormal{BS}$) & {$42$ dBm} \\
Maximum transmit power of the $n$th UE ($\overline{P}_n$) & {$23$ dBm} \\
{Minimum data rate requirement in UL ($\overline{R}^\textnormal{u}_n$)} & {$2$ bps$/$Hz} \\
{Minimum data rate requirement in DL ($\overline{R}^\textnormal{d}_n$)} & {$2$ bps$/$Hz} \\
{The penalty factor ($\lambda$)} & {$10^{\log (\frac{\overline{P}_\textnormal{BS}}{N_\textnormal{BS}})}$} \\
Channel realization number & $100$\\
\hline 
\end{tabular}
\end{table}

\subsection{Convergence of {Our} Proposed Algorithm}

\begin{figure}
\centering
\includegraphics[width=3.1in]{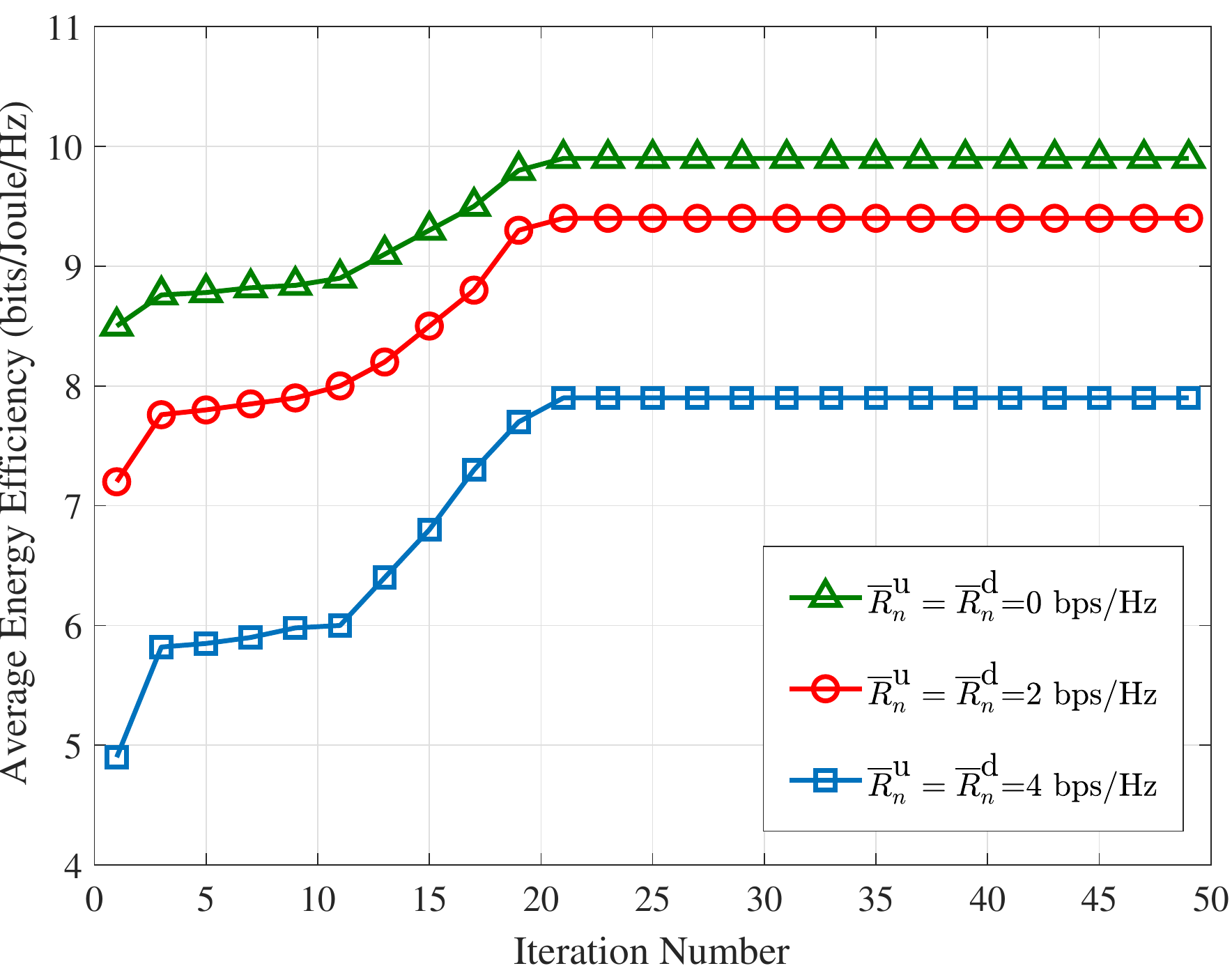}
\caption{{Average system EE (bits/Joule/Hz) versus {the} number of iterations for different levels of data rate requirements in UL and DL ($\overline{R}^\textnormal{u}_n$, $\overline{R}^\textnormal{d}_n$).}}
\label{fig_convergence}
\end{figure}

\begin{figure}
\centering
\includegraphics[width=3.1in]{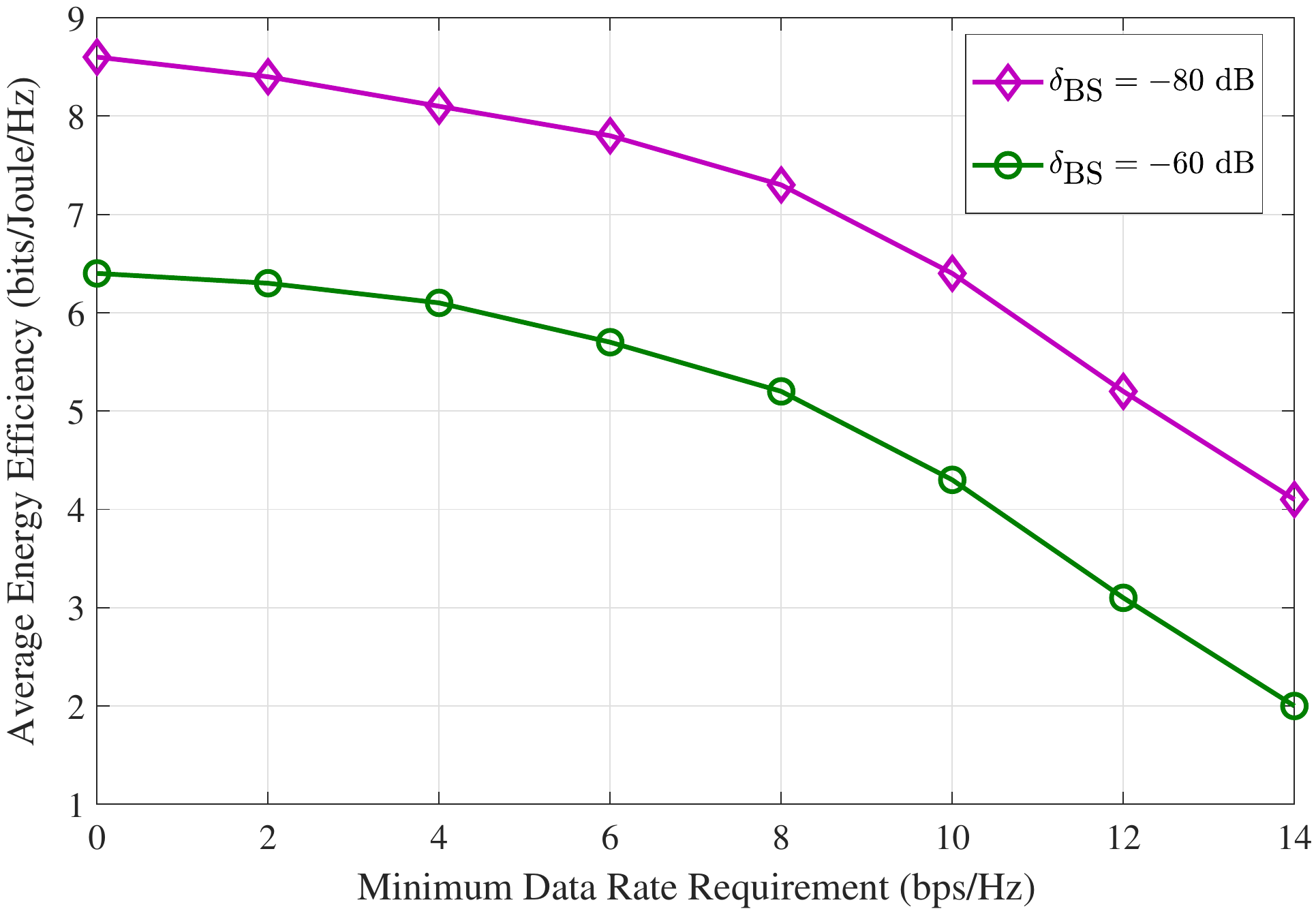}
\caption{{Average system EE (bits/Joule/Hz) versus {the} levels of data rate requirements in UL and DL ($\overline{R}^\textnormal{u}_n=\overline{R}^\textnormal{d}_n$) for different SIC constant {of {the} BS ($\delta_\textnormal{BS}$)}.}}
\label{fig_MinDataRate}
\end{figure}

First, we consider a system with different levels of data rate requirement{s} in UL and DL. Fig. \ref{fig_convergence} demonstrates the average system EE versus the number of iterations, for different minimum data rate requirements of {the} UEs in UL and DL ($\overline{R}^\textnormal{u}_n$, $\overline{R}^\textnormal{d}_n$). From Fig. \ref{fig_convergence}, it can be seen that the EE is monotonically non-decreasing function. Fig. \ref{fig_convergence} also shows the overall convergence of our proposed iterative algorithm. As observed in this figure, our proposed algorithm converges to a stationary point after 20 iterations. Hence, in the following case studies, we show the performance of our proposed algorithm for 20 iterations.
From Fig. \ref{fig_convergence}, we also see that increasing the minimum data rate requirement{s} of {the} UEs causes the average system EE to decrease. In fact, average EE has the highest value for the case of $\overline{R}^\textnormal{u}_n=\overline{R}^\textnormal{d}_n=0$ bps/Hz. 
The reason is that without the limitation on the data rate requirement for all UEs, the sub-carriers are assigned to the UEs with high channel gains in UL and DL which results in a higher system data-rate and a lower power consumption. Therefore, a higher average EE is achieved in a system without constraints on the QoS requirement for the UEs. To get in insight, we investigate the effect of the minimum data rate requirement on {the} EE as follows.

Fig. \ref{fig_MinDataRate} demonstrates the average system EE versus the levels of minimum data rate requirements in UL and DL, i.e., $\overline{R}^\textnormal{u}_n=\overline{R}^\textnormal{d}_n$, for different SIC constant of {the} BS ($\delta_\textnormal{BS}$). From this figure, we observe that increasing the minimum data rate requirement{s} of {the} UEs in UL and DL, i.e., $\overline{R}^\textnormal{u}_n$, $\overline{R}^\textnormal{d}_n$ leads to a decreased system EE. However, as $\overline{R}^\textnormal{u}_n$ and $\overline{R}^\textnormal{d}_n$ increase, the decline in the system EE becomes more significant. 
The reason is that as the minimum data rate requirement is low, the required transmit power to satisfy the QoS constraint is also low leading to lower total system power consumption and subsequently lower EE. In particular, as the minimum data rate requirement in UL and DL increases, the UEs and BS have to increase their transmit power in all of channels including low quality channels to satisfy the QoS requirement constraints in UL and DL, respectively. This results in a decreased system data rate and an increased total system power consumption which leads to a decreased system EE.

\subsection{{Probability of Feasibility}}

\begin{figure}
\centering
\includegraphics[width=3.1in]{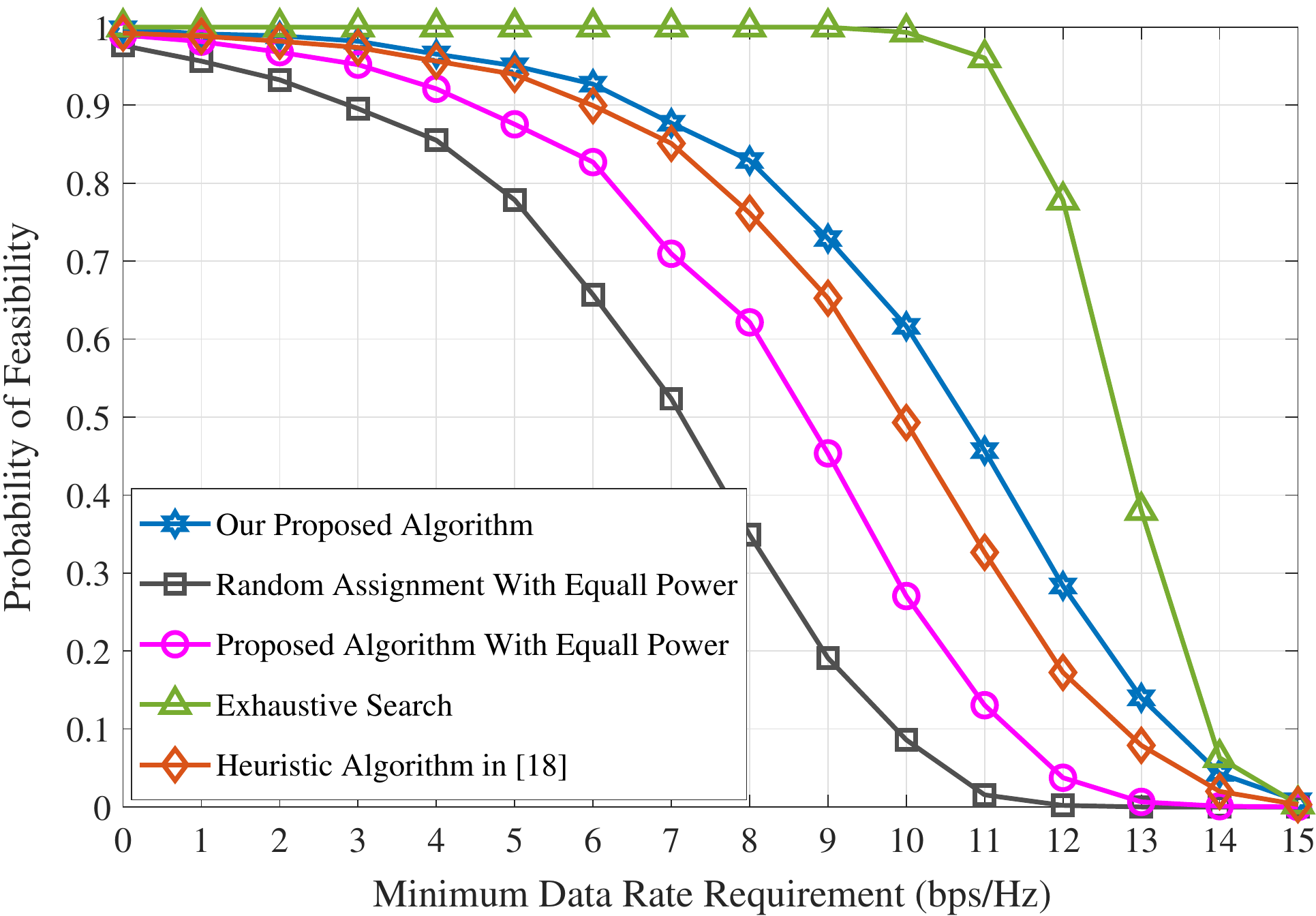}
\caption{{Probability of feasibility versus {the} levels of data rate requirements in UL and DL ($\overline{R}^\textnormal{u}_n=\overline{R}^\textnormal{d}_n$) for different algorithms ($N=2$ and $K=4$).}}
\label{fig_ProbFeasible}
\end{figure}

{Now, we investigate that for what levels of data rate requirements in UL and DL, the stated EE maximization problem is feasible. It is notable that for all simulation scenarios, if the UL and DL transmission rates for each UE do not meet the minimum data rate requirement in UL and DL, respectively, the optimization problem (\ref{convex problem}) is infeasible and so the total system sum-rate is equal to zero leading to zero system EE. To clarify more, we compare the probability of feasibility for our proposed algorithm with four different resource allocation schemes including the exhaustive search result. As the exhaustive search is a time consuming method, we consider a system with small number of UEs and sub-carriers, i.e., $N=2$ and $K=4$.
Fig. \ref{fig_ProbFeasible} demonstrates the probability of feasibility versus the levels of data rate requirements in UL and DL, i.e., $\overline{R}^\textnormal{u}_n=\overline{R}^\textnormal{d}_n$, for different algorithms. We compare our proposed algorithm with an algorithm with random sub-carrier assignment and equal power in which power in each assigned sub-carriers in UL and DL is equal to the total transmit power of UEs and BS divided by the number of assigned sub-carriers to the UEs and BS, respectively. We also compare our proposed algorithm with an algorithm with the same sub-carrier assignment with our algorithm but with equal power like previous algorithm. We compare our proposed algorithm with the heuristic algorithm proposed in \cite{wen} for {an} FD system with linear SI which first the power is allocated and then for the allocated power, the sub-carriers are assigned. We also compare our proposed algorithm with the exhaustive search result. 
From Fig. \ref{fig_ProbFeasible}, we observe that our proposed algorithm outperforms other algorithms except exhaustive search which finds the global optimal solution with high computational complexity. In fact, our proposed algorithm has better feasible set and its probability of feasibility is larger than the other ones. The reason is that our proposed algorithm obtains the joint sub-carrier assignment and power control policy and converges to a locally optimal solution which is a tight approximation for the optimal solution. This observation highlights the importance of joint resource allocation in improving the performance of IBFD cellular networks.}

\subsection{The Effect of SIC on {the} EE for Different {Cell Diameters}}

\begin{figure}
\centering
\includegraphics[width=3.2in]{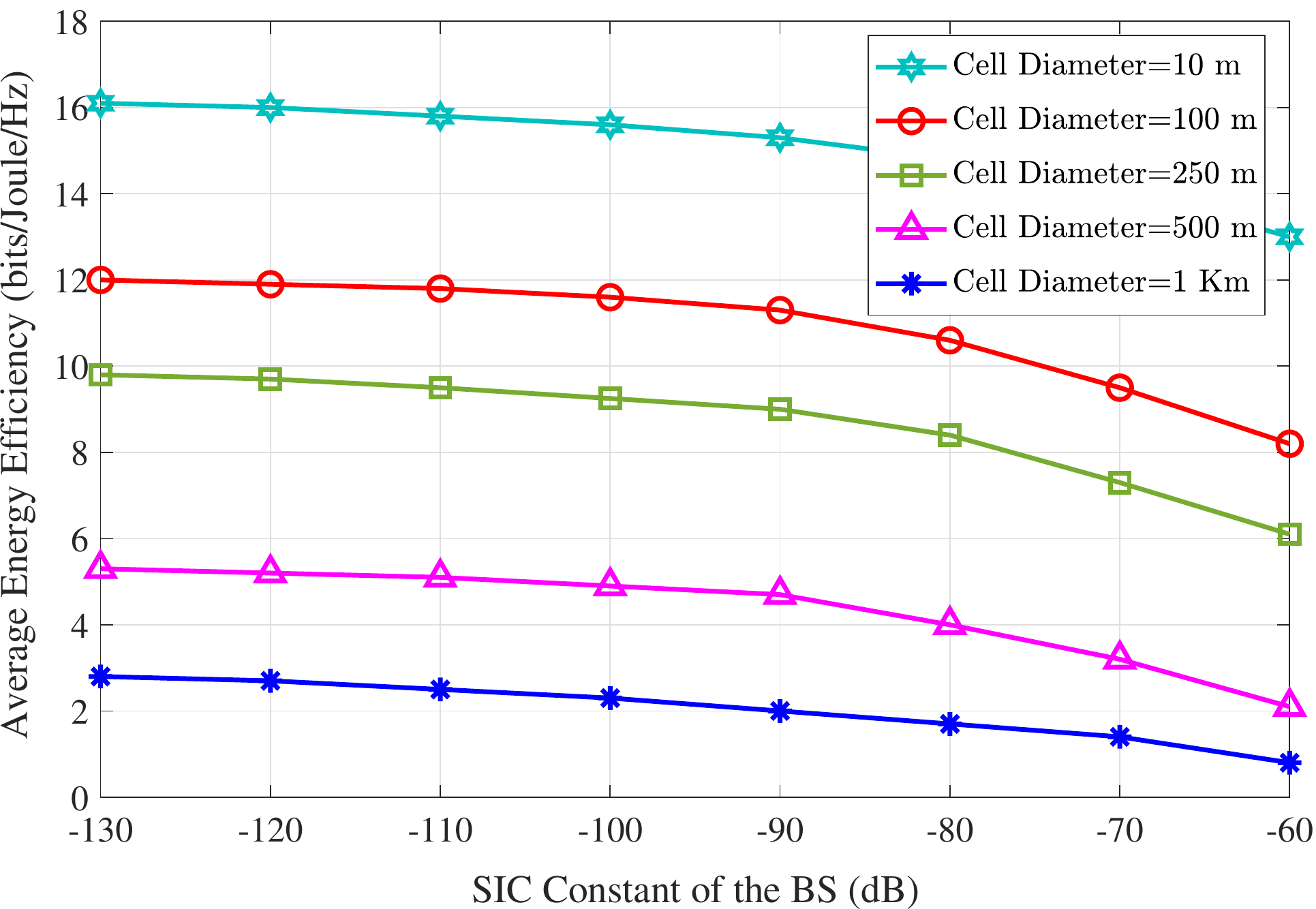}
\caption{{Average system EE (bits/Joule/Hz) versus SIC constant {of {the} BS ($\delta_\textnormal{BS}$)} for different {cell diameters} ({$\delta_n$=-70 dB}).}}
\label{fig_SIC_cellSideLength}
\end{figure}

In this scenario, we consider a system in which the SIC constant of {the} UEs ($\delta_n$) is set as -70 dB. Fig. \ref{fig_SIC_cellSideLength} demonstrates the average system EE versus the SIC constant of {the} BS for different {cell diameters}. We see that the average EE decreases with increasing the {cell diameter}. The reason is that as the {cell diameter} increases, the system data rate decreases and the transmit power{s} of both {the} UEs and BS increase {resulting} in {a} decreased system EE. From Fig. \ref{fig_SIC_cellSideLength}, we also observe {that} small values {of} the SIC constant {of {the} BS} results in an increased system EE. On the other hand, to achieve a certain amount of system EE, as the cell diameter increases, the SIC constant should be decreased which implies that {more efficient SIC methods} should be employed in long-distance communications. These results {suggest} that the IBFD capability is proper for short-distance communication such as small cell networks (i.e., femtocell and picocell networks) and D2D communications. Generally, the propriety of IBFD capability for short-distance communication has also been observed in \cite{IBFDforShortDistance1}, \cite{IBFDforShortDistance2} and \cite{IBFDforShortDistance3} for different system models and resource allocation problems. 
The same observations are also made for SIC constant of {the} UEs, the results of which are omitted due to similarity.

\subsection{{The Effect of Maximum Transmit Power on {the} EE}}

\begin{figure}
\centering
\includegraphics[width=3.2in]{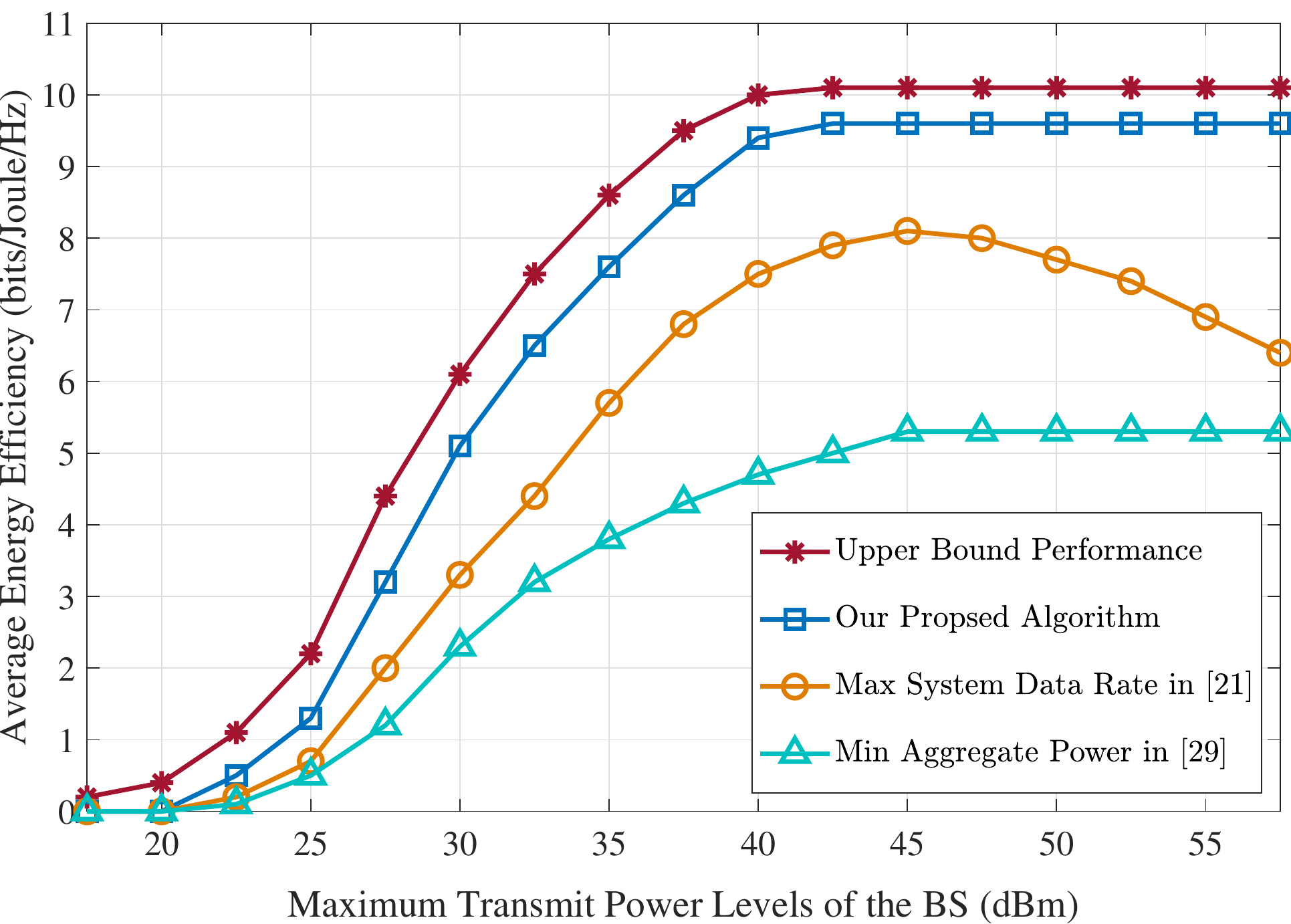}
\caption{{Average system EE (bits/Joule/Hz) versus the maximum transmit power levels of {the} BS ($\overline{P}_\textnormal{BS}$) for our proposed algorithm and baseline algorithm.}}
\label{fig_subcarrier_maxpow}
\end{figure}

Now, we consider a system in which the maximum transmit power of the $n$th UE is set as $\overline{P}_n=23$ dBm.  
Fig. \ref{fig_subcarrier_maxpow} shows the average system EE versus the maximum transmit power{s} of the BS ($\overline{P}_\textnormal{BS}$) for our proposed algorithm and two Baseline algorithms.
We observe that the average system EE by applying our proposed algorithm increases {by rising} maximum transmit power of the BS. In fact, the average system EE of our proposed algorithms is a monotonically non-decreasing function of the maximum transmit power. {However, the rate of the increase in EE is declined as the maximum transmit power becomes larger until the EE achieves a constant in the high transmit power regime. Particularly, starting from a small value of $\overline{P}_\textnormal{BS}$, the system EE first increases with increasing $\overline{P}_\textnormal{BS}$ and then saturates when $\overline{P}_\textnormal{BS} \geq 42$ dBm.} 
The reason is that applying our proposed algorithm, when the maximum system EE is obtained, a further increase in the maximum transmit power does not affect the EE.
In fact, although the higher values for transmit power increases the system total power consumption, {they} also cause the system data rate {to rise}. When the maximum available power of {the} BS is higher than certain levels, only a {portion} of {the power contributes to keeping} the EE at its maximum level, therefore the EE for higher value of {the} transmit power is constant.
In Fig. \ref{fig_subcarrier_maxpow}, the performance of upper bound is also illustrated. To find the upper bound, we consider a system with complete SIC capability for the BS and UEs and then we obtain the system EE for different maximum transmit power levels of {the} BS. From this figure, we see that our proposed algorithm achieves over 90\% of the upper bound performance. Also, as the optimal solution of problem (\ref{problem}) is located between the upper bound performance and the solution of our proposed algorithm, we can conclude that our proposed algorithm closely achieves the globally optimal solution. This observation highlights the tightness of the MM approximation used for making the stated problem convex in our paper.

Fig. \ref{fig_subcarrier_maxpow} also contains the average system EE of two other resource allocation algorithms: algorithms proposed in \cite{nam} for {the} system data-rate maximization in {an} FD system and \cite{Leng} for {the} aggregate power consumption minimization in {an} FD system.
It can be observed that our proposed algorithm outperforms two other algorithms proposed in \cite{nam} and \cite{Leng}. The reason is that the algorithm proposed in \cite{nam} uses excess power to increase the system data-rate by sacrificing EE, particularly in the high transmit power regime. On the other hand, the algorithm proposed in \cite{Leng} just considers the total system power consumption and after it reaches a certain system data-rate which satisfies the data rate constraints of UEs, it stops and does not use the executive power to increase the EE.

\subsection{Comparing the Performance of {Our} Proposed Algorithm with Existing Algorithms}

\begin{figure}
\centering
\includegraphics[width=3.2in]{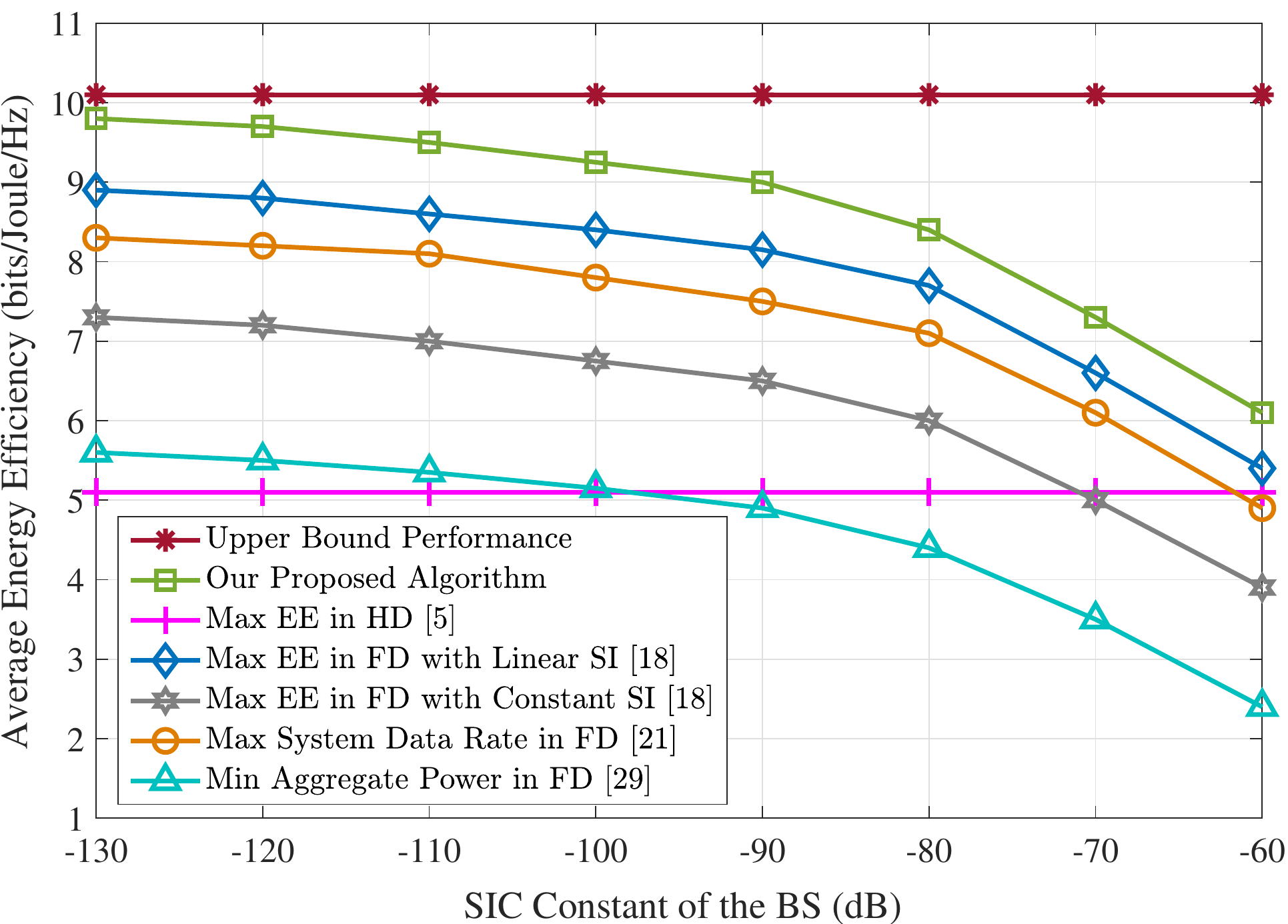}
\caption{Average system EE (bits/Joule/Hz) versus the SIC constant of the BS ($\delta_\textnormal{BS}$) for different algorithms ($\delta_n=-70$ dB).}
\label{fig_compare}
\end{figure}
Finally, we compare the performance of our proposed algorithm with other related algorithms. For this purpose, we consider five schemes proposed in \cite{He}, \cite{wen}, \cite{nam} and \cite{Leng} as well as the upper bound performance. As aforementioned, the stated problem for maximizing the EE in full-duplex (FD) system{s} with proportional SI has not been considered in the literature. Therefore, we compare our proposed algorithm with different scenarios available in the literature: we compare with \cite{He} for the EE maximization in an HD system, {with} \cite{wen} for the EE maximization in {an} FD system {with linear and constant SI,} {with} \cite{nam} for {the} system data-rate maximization in {an} FD system, with \cite{Leng} for {the} aggregate power consumption minimization in an FD system and with upper bound performmance which finds the optimal EE in a system with complete SIC capability for the BS and UEs (See Fig. \ref{fig_subcarrier_maxpow}). In these comparisons, the number of sub-carriers is set as $K=16$, which are assigned simultaneously to the UEs and BS for the UL and DL transmissions in FD systems (i.e., for our proposed algorithm, upper bound performance and the algorithms proposed in \cite{wen}, \cite{nam} and \cite{Leng}), while they are assigned only to the BS for the DL transmission in HD systems (i.e., for the algorithm proposed in \cite{He}). As in reality, we assume that the SI exists and it is proportional to the transmit power with Rician distribution for SI channels in our FD system model. In the proposed algorithm in \cite{wen}, to maximize the EE in FD systems with constant SI model, the SI was constant. Therefore, to carry out a fair comparison with our proposed algorithm, we include the SI propotional to the transmit power in data rate calculations of \cite{wen} with constant SI model. 
We also compare our proposed algorithm with heauristic algorithm proposed in \cite{wen} for linear SI model.
Fig. \ref{fig_compare} demonstrates the average system EE versus the SIC constant of the BS for seven different scenarios: the problem of EE maximization in FD systems with proportional SI (i.e., our proposed algorithm), the problem of EE maximization in HD systems \cite{He}, the problem of EE maximization in FD systems with linear SI \cite{wen}, the problem of EE maximization in FD systems with constant SI \cite{wen}, the problem of system data rate maximization in FD systems \cite{nam}, the problem of aggregate power minimization in FD systems \cite{Leng} and the upper bound performance. As observed in Fig. \ref{fig_compare}, our proposed algorithm outperforms the other algorithms proposed in the literature and it has the closest performance to the upper bound\footnote{{Since the upper bound performance finds the optimal EE in a system with complete SIC for the BS and UEs, its achieved EE is independent from the value of the SIC constant of the BS.}}.
More specifically, the EE achieved by our proposed algorithm {that maximizes} EE with proportional SI is 
more than that achieved by the algorithms proposed to maximize the system data rate and minimize the aggregate power. The reason is that the algorithms proposed to maximize the system data rate and minimize the aggregate power just optimize the system data rate and total power consumption, respectively while our proposed algorithm optimizes both of them which leads to a higher EE.
In addition, as observed in Fig. \ref{fig_compare}, our proposed algorithm performs 20\% and 35\% better than the algorithms proposed to maximize EE in FD system{s} with linear and constant SI model, respectively. The reason is that the algorithm proposed in \cite{wen} for linear SI model is a heuristic algorithm which decouples the problem into two sub-problems of power control and sub-carrier assignment, while our proposed algorithm finds sub-carrier assignment and power control policy jointly and it closely achieves the globally optimal solution (See Fig. \ref{fig_subcarrier_maxpow}). In addition, although the algorithm proposed in \cite{wen} for constant SI model finds an optimal solution, it does not have a good performance in a realistic system where the SI is proportional to the transmit power.
The important observation from Fig. \ref{fig_compare} is that by applying the IBFD capability with efficient SIC techniques, i.e. $\delta_\textnormal{BS} \leq -100$ dB and $\delta_n \leq -70$ dB, our resource allocation scheme can operate 75\% more energy efficiently than that in HD system{s}\footnote{Since the proposed algorithm in \cite{He} {that} maximizes {the} EE operates in HD system{s}, its achieved EE is independent from the value of SIC constant of the BS.}. The reason is that although algorithms that maximize EE in HD systems achieve a lower power consumption, a two-fold data rate is achieved by algorithms that maximize EE in FD systems. The same results are also obtained for SIC constant of {the} UEs which {omitted} due to similarity.

\section{Conclusions}
We have studied the EE maximization problem subject to the maximum transmit power of {the} BS and UEs, {while} satisfying QoS requirements of the UEs at UL and DL in OFDMA IBFD network{s}. We addressed this problem by proposing an algorithm for jointly optimizing {the} power control and sub-carrier assignment in {the} UL and DL. As the formulated optimization problem {was} non-convex, solving it in original form was difficult. Thus, we first reformulated it in a problem by decoupling the transmit power and sub-carrier assignment variables. Then by applying Dinkelbach method, we transformed the fractional objective function to a subtractive form. Next, we made the inner problem in each iteration of Dinkelbach algorithm convex by using MM algorithm and handled the integer sub-carrier assignment variables by applying \emph{abstract Lagrangian duality} and introducing a penalty function. We showed that MM approximation used for making the inner problem convex is a tight lower band of the original inner problem. Finally, we proposed an iterative resource allocation scheme to solve the inner problem which converges to the locally optimal solution. Simulation results {showed} that our proposed algorithm quickly converges and outperforms current schemes in the literature such as the algorithm proposed to minimize the aggregate power or maximize the system data rate. Also, our simulation results demonstrated that by applying the IBFD capability in cellular networks with efficient SIC techniques, the resource allocation scheme can operate 75\% more energy efficiently than that in {an} HD system. An interesting future work includes considering the case of multi-antenna BS which requires joint beamforming design, sub-carrier assignment and power control solutions.

\appendix
\addcontentsline{toc}{section}{Appendices}
\renewcommand{\thesubsection}{\Alph{subsection}}

{
\subsection{Proof of Proposition \ref{proposition1}}\label{appendix A}
We start the proof of Proposition \ref{proposition1} by using the \emph{abstract Lagrangian duality}.~The primal problem of (\ref{new problem}) is written as 
\begin{equation}\label{30}
p^{*}=\max_{(\widetilde{\bold{p}}, \bold{x}) \in \mathcal{D}}\min_{\lambda}L(\bold{x},\widetilde{\bold{p}},\lambda),
\end{equation}
where the dual problem of the (\ref{new problem}) is given by
\begin{equation}\label{31}
d^{*}=\min_{\lambda}\max_{(\widetilde{\bold{p}}, \bold{x}) \in \mathcal{D}}L(\bold{x},\widetilde{\bold{p}},\lambda)\triangleq \min_{\lambda}\theta(\lambda),
\end{equation}
where $\theta(\lambda)$ is defined as $\displaystyle{\theta(\lambda)\triangleq \max_{(\widetilde{\bold{p}}, \bold{x}) \in \mathcal{D}}L(\bold{x}, \widetilde{\bold{p}})}.$
Based on the weak duality theorem, we have the following equality:
\begin{equation}\label{32}
p^{*}=\max_{(\widetilde{\bold{p}}, \bold{x}) \in \mathcal{D}} \min_{\lambda}L(\bold{x},\widetilde{\bold{p}},\lambda) \leq \min_{\lambda}\theta(\lambda)=d^{*}.
\end{equation}
It should be noted that for $\widetilde{\bold{p}}, \bold{x} \in \mathcal{D}$, we have two cases where each case should be studied. \\
\textbf{Case 1}: Assume that at the optimal point, we have:
\begin{equation}\label{33}
\sum_{\forall n \in \mathcal{N}}{\sum_{\forall k \in \mathcal{K}}{\left( x_{n,k} - (x_{n,k})^2 \right)}}=0.
\end{equation}
In this case, $d^{*}$ is also a feasible solution of (\ref{new problem}).~Accordingly,~substituting the optimal value of $\lambda$,~i.e., $\lambda^{*}$,~into the optimization problem (\ref{new problem}) leads to the following equation:
\begin{equation}\label{34}
d^{*}= \theta(\lambda^{*})=\max_{(\widetilde{\bold{p}}, \bold{x}) \in \mathcal{D}} \min_{\lambda}L(\bold{x},\widetilde{\bold{p}},\lambda)=p^{*}.  
\end{equation}
Furthermore,~referring to (\ref{lagrangian function}), in this region $\theta(\lambda)$ is a monotonically decreasing function with respect to $\lambda$.~On the other hand,~it is specified that $d^{*}=\min_{\lambda}\theta(\lambda)$.~Hence,~we have:
\begin{equation}\label{35}
d^{*}=\theta(\lambda), ~\forall \lambda\geq \lambda^{*}. 
\end{equation}
Equation (\ref{35}) indicates that for any value of $\lambda\geq \lambda^{*}$, the solution of (\ref{new problem with penalty term}) leads the optimal solution of (\ref{new problem}).\\
\textbf{Case 2}: Assume that, $x_{n,k}$ take values $0<x_{n,k}<1$,~causing:
\begin{equation}
  \sum_{\forall n \in \mathcal{N}}{\sum_{\forall k \in \mathcal{K}}{\left( x_{n,k} - (x_{n,k})^2 \right)}}>0.  
\end{equation}
In this case, referring to (\ref{lagrangian function}) and $\theta(\lambda)$,~at the optimal point $\theta(\lambda^*)$ tends to $-\infty$.~However,~this may not happen as it contradicts with primal solution (i.e., $\max-\min$ inequality) which states that $\theta(\lambda^*)$ is bounded from below by solution of (\ref{new problem}) which is always greater than zero.~Thus,~at the optimal point, we have $\sum_{\forall n \in \mathcal{N}}{\sum_{\forall k \in \mathcal{K}}{\left( x_{n,k} - (x_{n,k})^2 \right)}}=0$,~and the result for the first case is hold.~This completes the proof.}

{
\subsection{Proof of Proposition \ref{proposition2}}\label{appendix B}
The approximation (\ref{taylor approximation}) makes a tight lower bound of $\bold{L}(\mathbf{x},\widetilde{\bold{p}})$. 
Since $e_2(\bold{x},\widetilde{\bold{p}})$ is a concave function,~the gradient of $e_2(\bold{x},\widetilde{\bold{p}})$ is supper-gradient \cite{Boyd} as follows:
\begin{equation}
e_2(\bold{x},\widetilde{\bold{p}})\leq \widetilde{e}_2(\bold{x}, \widetilde{\bold{p}}) .
\end{equation}
It is noteworthy that $e_1(\bold{x},\widetilde{\bold{p}})-e_2(\bold{x},\widetilde{\bold{p}}) \geq e_1(\bold{x},\widetilde{\bold{p}})-\widetilde{e}_2(\bold{x}, \widetilde{\bold{p}})$.~Moreover,~the equality holds when $\bold{x}=\bold{x}^{(t-1)}$ and $\widetilde{\bold{p}}=\widetilde{\bold{p}}^{(t-1)}$ which shows the tightness of the lower bound.~In addition,~we can conclude that the solution obtained by incorporating MM approximation would be improved at the end of each iteration.
The objective function of (\ref{new problem with penalty term}) in $t$-th iteration is $e_1(\bold{x}^{(t)},\widetilde{\bold{p}}^{(t)})-e_2(\bold{x}^{(t)},\widetilde{\bold{p}}^{(t)})$. Hence, we have the following equation:
\begin{align}
&e_1(\bold{x}^{(t+1)},\widetilde{\bold{p}}^{(t+1)})-e_2(\bold{x}^{(t+1)},\widetilde{\bold{p}}^{(t+1)})\geq e_1(\bold{x}^{(t+1)},\widetilde{\bold{p}}^{(t+1)})\nonumber\\
&- e_2(\bold{x}^{(t)},\widetilde{\bold{p}}^{(t)}) - \nabla_{\widetilde{\bold{p}}}{e_2^T\!\big(\bold{x}^{(t)}\!,\!\widetilde{\bold{p}}^{(t)}\big)} \!\big(\widetilde{\bold{p}}^{(t+1)}\!-\!\widetilde{\bold{p}}^{(t)}\big) \nonumber\\
&-\!\!\nabla_{\bold{x}}e_2^T\!\big(\bold{x}^{(t)}\!,\!\widetilde{\bold{p}}^{(t)}\!\big)\! \big(\bold{x}^{(t+1)}\!\!-\!\bold{x}^{(t)}\!\big)\!\!=\!\!\max_{\widetilde{\bold{p}},\!\bold{x}}e_1(\bold{x},\widetilde{\bold{p}})\!-\!e_2(\bold{x}^{(t)},\widetilde{\bold{p}}^{(t)})\nonumber\\
&-\nabla_{\widetilde{\bold{p}}}{e_2^T\!\big(\bold{x}^{(t)}\!,\!\widetilde{\bold{p}}^{(t)}\big)}\!\big(\widetilde{\bold{p}}^{}\!-\!\widetilde{\bold{p}}^{(t)}\big)-\nabla_{\bold{x}}e_2^T\big(\bold{x}^{(t)}, \widetilde{\bold{p}}^{(t)}\big)\big(\bold{x}^{}-\bold{x}^{(t)}\big)\nonumber\\
&\geq \!e_1(\bold{x}^{(t)},\widetilde{\bold{p}}^{(t)})\!-\!e_2(\bold{x}^{(t)},\widetilde{\bold{p}}^{(t)})\!-\!\nabla_{\widetilde{\bold{p}}}{e_2^T\!\big(\bold{x}^{(t)}\!,\!\widetilde{\bold{p}}^{(t)}\big)}\!\big(\widetilde{\bold{p}}^{(t)}\!-\!\widetilde{\bold{p}}^{(t)}\big)\nonumber\\
&-\!\nabla_{\bold{x}}e_2^T\big(\bold{x}^{(t)},\!\widetilde{\bold{p}}^{(t)}\big)\big(\bold{x}^{(t)}\!\!-\!\bold{x}^{(t)}\big)\!=\!e_1(\bold{x}^{(t)},\!\widetilde{\bold{p}}^{(t)})\!-\!e_2(\bold{x}^{(t)},\!\widetilde{\bold{p}}^{(t)})
\end{align}
Thus, by solving the convex lower bound in (\ref{convex problem}),~the proposed iterative algorithm generates a sequence of feasible solutions, i.e., $\widetilde{\bold{p}}^{(t+1)}$ and $\bold{x}^{(t+1)}$.~One may conclude that the solution of (\ref{new problem with penalty term}) would be improved and takes larger values as iterations continue which yields a locally optimal solution.}

\bibliographystyle{IEEEtran}
\bibliography{Mybib}


\bibliographystyle{IEEEtran}

\begin{IEEEbiography}[{\includegraphics[width=1in,height=1.25in,clip,keepaspectratio]{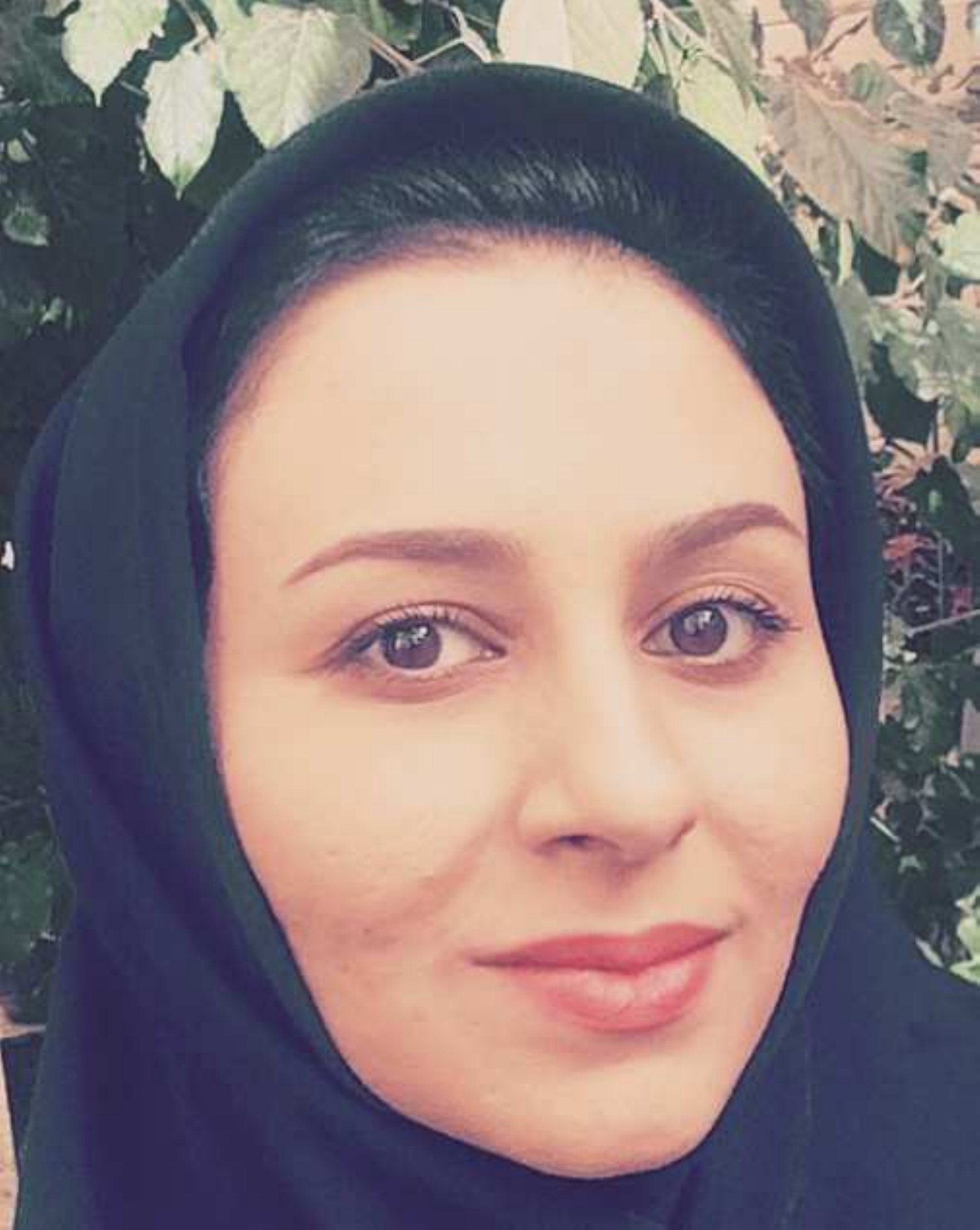}}]
{Rojin Aslani} (S'19) received her B.Sc. degree in Information Technology Engineering from University of Tabriz, Tabriz, Iran, in 2011 and her M.Sc. degree in Information Technology Engineering (Computer Networks) from Amirkabir University of Technology, Tehran, Iran, in 2013. She is pursuing the Ph.D. degree in Computer Engineering (Computer Networks) in Amirkabir University of Technology, Tehran, Iran. Currently, she is a visiting research scholar at the Department of Electrical and Computer Engineering, University of Nevada, Las Vegas, NV, USA. Her current research area include resource allocation in wireless networks, vehicular communications, and optimization.
\end{IEEEbiography}

\vspace{-25px}

\begin{IEEEbiography}[{\includegraphics[width=1in,height=1.25in,clip,keepaspectratio]{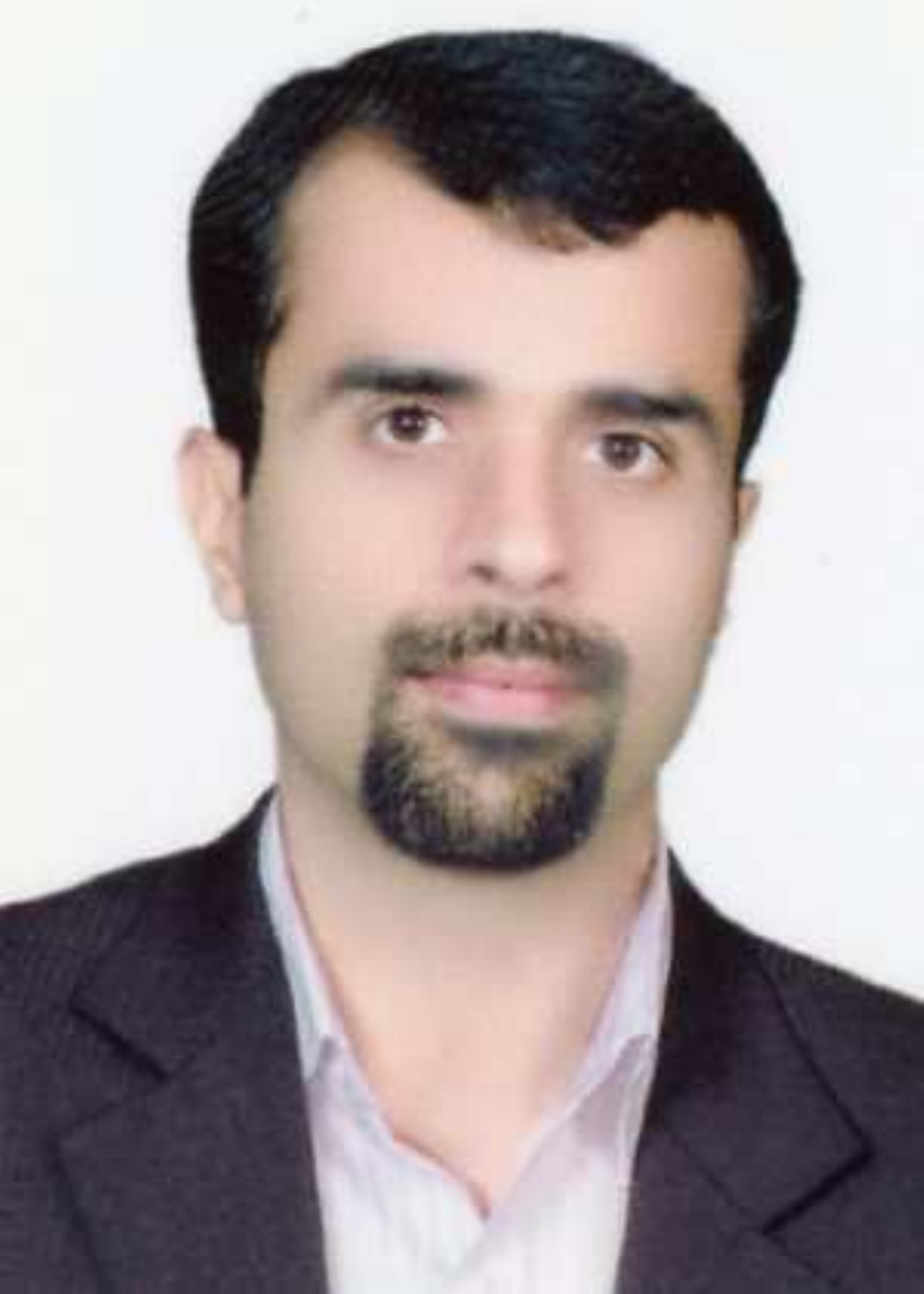}}]
{Mehdi Rasti} (S'08-M'11) received his B.Sc. degree from Shiraz University, Shiraz, Iran, and the M.Sc. and Ph.D. degrees both from Tarbiat Modares University, Tehran, Iran, all in Electrical Engineering in 2001, 2003 and 2009, respectively. 
From November 2007 to November 2008, he was a visiting researcher at the Wireless@KTH, Royal Institute of Technology, Stockholm, Sweden. From September 2010 to July 2012 he was with Shiraz University of Technology, Shiraz, Iran, after that he joined the Department of Computer Engineering and Information Technology, Amirkabir University of Technology, Tehran, Iran, where he is now an assistant professor. 
From June 2013 to August 2013, and from July 2014 to August 2014 he was a visiting researcher in the Department of Electrical and Computer Engineering, University of Manitoba, Winnipeg, MB, Canada. His current research interests include radio resource allocation in wireless networks and network security.
\end{IEEEbiography}

\vspace{-25px}

\begin{IEEEbiography}[{\includegraphics[width=1in,height=3.25in,clip,keepaspectratio]{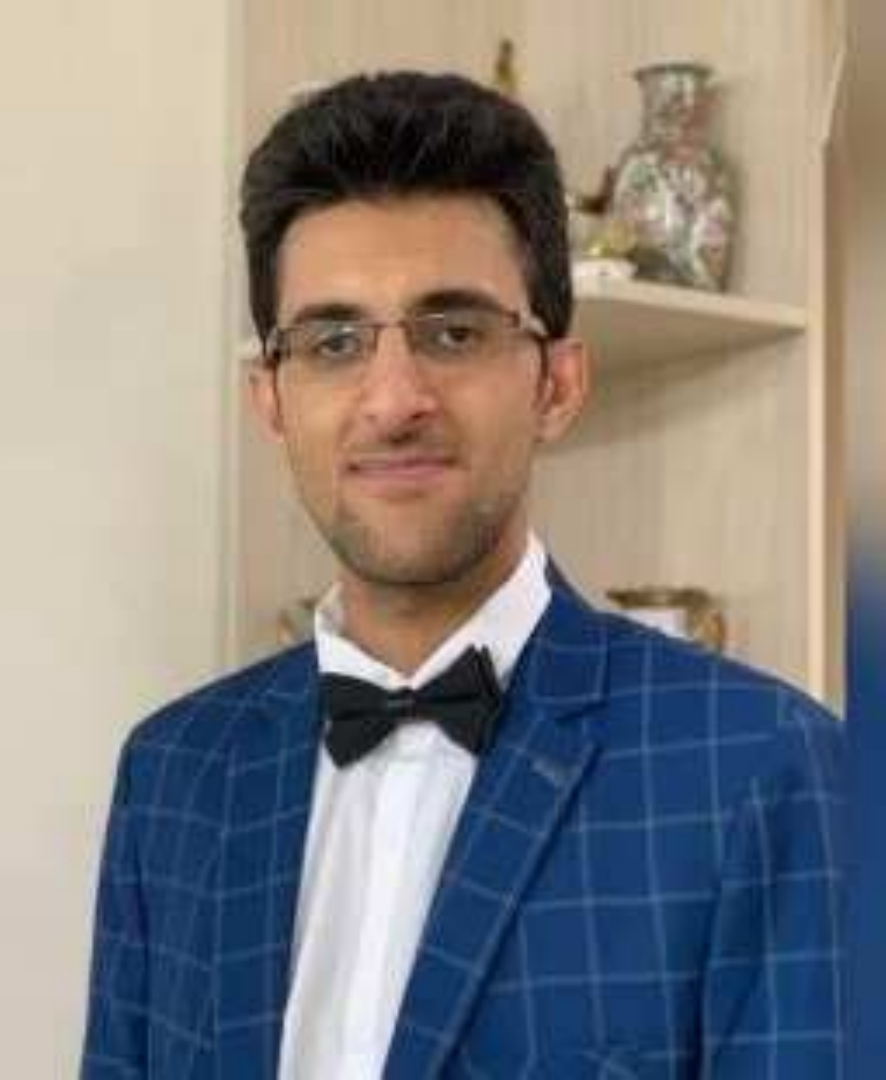}}]
{Ata Khalili} (S’18) received the B.Sc. degree and M.Sc. degree with first class honors in Electronic Engineering and Telecommunication Engineering from the Shahed University in 2016 and 2018, respectively. He is now working as a visiting researcher in the Department of Computer Engineering and Information Technology, Amirkabir University of Technology, Tehran, Iran. His research interests include convex and non-convex optimization, resource allocation in wireless communication, Green communication, and mobile edge computing.
\end{IEEEbiography}

\end{document}